\documentclass{article}
\usepackage{fullpage}
\usepackage{amsmath,amssymb,amsthm}
\usepackage{tikz}

\newtheorem{theorem}{Theorem}
\newtheorem{lemma}{Lemma}
\newtheorem{proposition}{Proposition}
\newtheorem{corollary}{Corollary}

\newcommand{\valset}[1]{V_{{#1}}}
\newcommand{\rvalset}[2]{R_{{#1},{#2}}}
\newcommand{\combset}[2]{C_{{#1},{#2}}}
\newcommand{\altern}{a}

\newcommand{\mech}[2]{{\mbox{\rm Mech}}_{{#1},{#2}}}
\newcommand{\rmech}[3]{{\mbox{\rm Mech}}^{\mbox{\bf\tiny{#1}}}_{{#2},{#3}}}
\newcommand{\mratio}{\mbox{\rm ratio}}
\newcommand{\ratio}[2]{r_{{#1},{#2}}}
\newcommand{\aratio}[1]{r_{{#1}}}
\newcommand{\gratio}[3]{r^{\mbox{\bf\tiny{{#1}}}}_{{#2},{#3}}}
\newcommand{\agratio}[2]{r^{\mbox{\bf\tiny{{#1}}}}_{{#2}}}
\newcommand{\unilat}[3]{{U^{{#3}}_{{#1},{#2}}}}
\newcommand{\duple}[3]{{D^{{#3}}_{{#1},{#2}}}}

\newcommand{\oul}{0.610}
\newcommand{\ouu}{0.611}
\newcommand{\ogl}{0.616}
\newcommand{\ogu}{0.641}
\newcommand{\cul}{0.660}
\newcommand{\cgl}{0.660}
\newcommand{\cgu}{0.750}
\newcommand{\cone}{c_1}
\newcommand{\ctwo}{c_2}

\begin{document}
% Page heads

\title{Truthful Approximations to Range Voting\thanks{The authors acknowledge support from the Danish National Research Foundation and The National Science
Foundation of China (under the grant 61061130540) for the Sino-Danish Center for the Theory of Interactive Computation,
within which this work was performed. The authors also acknowledge support from the Center for Research
in Foundations of Electronic Markets (CFEM), supported by the Danish Strategic Research Council.
Author's addresses: Department of Computer Science, Aarhus University, Aabogade 34, DK-8200 Aarhus N, Denmark.}
}
\author{
Aris Filos-Ratsikas \and Peter Bro Miltersen
}

%\category{C.2.2}{Computer-Communication Networks}{Network Protocols}
%\terms{Design, Algorithms, Performance}
%\keywords{Wireless sensor networks, media access control,
%multi-channel, radio interference, time synchronization}

%\acmformat{Zhou, G., Wu, Y., Yan, T., He, T., Huang, C., Stankovic,
%J. A., and Abdelzaher, T. F.  2010. A multifrequency MAC specially
%designed for  wireless sensor network applications.}

\maketitle
\begin{abstract}
We consider the fundamental mechanism design problem of {\em
approximate social welfare maximization} under {\em general cardinal
preferences} on a finite number of alternatives and {\em without
money}. The well-known {\em range voting} scheme can be thought of as
a non-truthful mechanism for exact social welfare maximization in this
setting.  With $m$ being the number of alternatives, we exhibit a
randomized truthful-in-expectation ordinal mechanism implementing an outcome
whose expected social welfare is at least an $\Omega(m^{-3/4})$ fraction of
the social welfare of the socially optimal alternative. On the other hand, we 
show that for sufficiently many agents and any
truthful-in-expectation ordinal mechanism, there is a valuation profile where
the mechanism achieves at most an $O(m^{-2/3})$ fraction of the
optimal social welfare in expectation. Furthermore, we prove that no truthful-in-expectation (not necessarily ordinal)
mechanism can achieve 0.94-fraction of the optimal social welfare. We get tighter bounds for the natural
special case of $m = 3$, and in that case furthermore obtain separation results concerning the approximation
ratios achievable by natural restricted classes of
truthful-in-expectation mechanisms. In particular, we show that for
$m=3$ and a sufficiently large number of agents, the best mechanism that is {\em ordinal} as well as 
{\em mixed-unilateral} has an approximation ratio between \oul\ 
and \ouu, the best {\em ordinal}
mechanism has an approximation ratio between \ogl\ and \ogu, while the
best {\em mixed-unilateral} mechanism has an approximation ratio bigger than\ \cul. In particular, the best mixed-unilateral non-ordinal (i.e., cardinal) mechanism strictly
outperforms all ordinal ones, even the non-mixed-unilateral ordinal ones.
\end{abstract}
\section{Introduction}
We consider the fundamental mechanism design problem of {\em
approximate social welfare maximization} under {\em general cardinal
preferences} and {\em without money}. In this setting, there is a
finite set of agents (or {\em voters}) $N = \{1,\ldots,n\}$ and a finite set of
alternatives (or {\em candidates}) $M = \{1,\ldots,m\}$.  Each voter $i$ has a private
valuation function $u_i: M \rightarrow {\mathbb{R}}$ that can be
arbitrary, except that we require\footnote{We make this requirement primarily for convenience; to avoid having to qualifiy in technically annoying ways a number of definitions and statements of this paper as well as definitions and statements of previous ones.} that it is injective, i.e., we insist that it induces a total order on candidates. Standardly, the function $u_i$ is considered well-defined only up to positive affine
transformations. That is, we consider $x \rightarrow a u_i(x) + b$, for
$a  > 0$ and any $b$, to be a different representation of
$u_i$. Given this, we fix the representative $u_i$ that maps the
least preferred candidate of voter $i$ to $0$ and the most preferred
candidate to $1$ as the canonical representation of $u_i$ and we
shall assume that all $u_i$ are thus canonically represented throughout 
this paper. In particular, we shall let $\valset{m}$
denote the set of all such functions.

We shall be interested in
{\em direct revelation mechanisms without money} that elicit the {\em valuation profile} ${\bf u} = (u_1,
u_2, \ldots, u_n)$ from the voters and based on this elect a candidate $J({\bf u}) \in
M$. We shall allow mechanisms to be randomized and $J({\bf u})$ is
therefore in general a random map.  In fact, we shall define a mechanism simply to be a random map $J: {\valset{m}}^n \rightarrow M$. We prefer mechanisms that
are {\em truthful-in-expectation}, by which we mean that the following condition is satisfied: For each voter $i$, and all ${\bf u} = (u_i, u_{-i}) \in
{\valset{m}}^n$ and $\tilde u_i \in \valset{m}$, we have 
$E[u_i(J(u_i,u_{-i})] \geq E[u_i(J(\tilde u_i, u_{-i})]$. 
That is, if voters are assumed to be expected utility maximizers, the
optimal behavior of each voter is always to reveal their true valuation
function to the mechanism.  As truthfulness-in-expectation is the only notion of truthfulness of interest to us in this paper, we shall use ``truthful'' as a synonym for ``truthful-in-expectation'' from now on. Furthermore, we are interested in
mechanisms for which the expected {\em social welfare}, i.e.,
$E[\sum_{i=1}^n u_i(J({\bf u}))]$, is as high as possible, and we shall in
particular be interested in the {\em approximation ratio} $\mratio(J)$ of the
mechanism, defined by
\[ \mratio(J) = \inf_{{\bf u} \in {\valset{m}}^n} \frac{E[\sum_{i=1}^n u_i(J({\bf u}))]}{\max_{j \in M}\sum_{i=1}^n u_i(j)}, \]
trying to achieve mechanisms with as high an approximation ratio as possible.
Note that for $m=2$, the problem is easy; a majority vote is a truthful mechanism that achieves optimal social welfare, i.e., it has approximation ratio 1, so we only consider the problem for $m \geq 3$. 

A mechanism without money for general cardinal preferences can be
naturally interpreted as a {\em cardinal voting scheme} in which each
voter provides a {\em ballot} giving each candidate $j \in M$ a
numerical score between 0 and 1. A winning candidate is then determined
based on the set of ballots. With this interpretation, the well-known
{\em range voting scheme} is simply the determinstic mechanism that elects the
socially optimal candidate argmax$_{j \in M}\sum_{i=1}^n u_i(j)$, or, more precisely, elects this candidate {\em if} ballots are reflecting the true 
valuation functions $u_i$. In particular, range voting has by construction an approximation ratio of 1. However, range voting is not a 
truthful mechanism.

Before stating our results, we mention for comparison the approximation ratio of some simple truthful mechanisms. Let {\em random-candidate} be the mechanism that elects a candidate uniformly at random, without looking at the ballots. Let {\em random-favorite} be the mechanism that picks a voter uniformly at random and elects his favorite candidate; i.e., the (unique) candidate to which he assigns valuation $1$. Let {\em random-majority} be the mechanism that picks two candidates uniformly at random and elects one of them by a majority vote. It is not difficult to see that as a function of $m$ and assuming that $n$ is sufficiently large, {\em random-candidate} as well as {\em random-favorite} have approximation ratios $\Theta(m^{-1})$, so this is the trivial bound we want to beat. Interestingly, {\em random-majority} performs even worse, with an approximation ratio of $\Theta(m^{-2})$.

As our first main result, we exhibit a randomized
truthful mechanism with an approximation ratio of $0.37
m^{-3/4}$. The mechanism is the following very simple one: {\em With
probability $3/4$, pick a candidate uniformly at random. With
probability $1/4$, pick a random voter, and pick a candidate uniformly
at random from his $\lfloor m^{1/2} \rfloor$ most preferred
candidates.} Note that this mechanism is {\em ordinal}: Its behavior
depends only on the {\em rankings} of the candidates on the ballots,
not on their numerical scores. We know no asymptotically better
truthful mechanism, even if we allow general (cardinal)
mechanisms, i.e., mechanisms that can depend on the numerical scores
in other ways.  We also show a negative result: For sufficiently many
voters and any truthful ordinal mechanism, there is a
valuation profile where the mechanism achieves at most an
$O(m^{-2/3})$ fraction of the optimal social welfare in expectation.
The negative result also holds for non-ordinal mechanisms that are {\em mixed-unilateral}, by which we mean mechanisms that elect a candidate based on the ballot of a single randomly chosen voter.
%in fact holds not just for ordinal mechanisms, but
%for a wider class we call ``regular mechanisms'' and which might
%possibly contain all truthful-in-expectation mechanisms (for
%definitions and discussion, see below).

We get tighter bounds for the natural case of $m = 3$ candidates and
for this case, we also obtain separation results concerning the
approximation ratios achievable by natural restricted classes of
truthful mechanisms. Again, we first state the
performance of the simple mechanisms defined above for comparison: For
the case of $m=3$, {\em random-favorite} and {\em random-majority}
both have approximation ratios $1/2 + o(1)$ while {\em
random-candidate} has an approximation ratio of $1/3$.  We show that
for $m=3$ and large $n$, the best mechanism that is {\em ordinal} as
well as {\em mixed-unilateral} has an approximation ratio between \oul\
and \ouu. The best {\em ordinal} mechanism has an approximation ratio
between \ogl\ and \ogu. Finally, the best {\em mixed-unilateral}
mechanism has an approximation ratio larger than \cul. In particular,
the best mixed-unilateral mechanism strictly outperforms all
ordinal ones, even the non-unilateral ordinal ones. The
mixed-unilateral mechanism that establishes this is a convex
combination of {\em quadratic-lottery}, a mechanism of Feige and Tennenholtz \cite{Feige10}
and {\em random-favorite}, that was defined above.
\subsection{Background, related research and discussion}
Characterizing strategy-proof social choice functions (a.k.a.,
truthful direct revelation mechanisms without money) under general
preferences is a classical topic of mechanism design and social choice
theory. The celebrated Gibbard-Satterthwaite theorem
\cite{Gibbard73,Satterthwaite75} states that when the number $m$ of candidates is at least 3, any {\em deterministic} and {\em onto} truthful mechanism\footnote{Even though the theorem is usually stated for ordinal mechanisms, it is easy to see that it holds even without assuming that the mechanism is ordinal.} must be a {\em dictatorship},
i.e., it is a function of the ballot of a single distinguished voter
only, and outputs the favorite (i.e., top ranking) candidate of that
voter. Gibbard \cite{Gibbard77} extended the Gibbard-Satterthwaite theorem to the case
of randomized ordinal mechanisms, and we shall heavily use his theorem when proving our negative results on ordinal mechanisms:
\begin{theorem}\cite{Gibbard77}\label{thm-Gibbard77}
The ordinal mechanisms without money that are truthful under general cardinal preferences\footnote{without ties, i.e., valuation functions must be injective, as we require throughout this paper, except in Theorem \ref{thm:anyupper}. If ties were allowed, the characterization would be much more complicated.} are exactly the convex combinations of truthful {\em unilateral} ordinal mechanisms and truthful {\em duple} mechanisms.
\end{theorem}
Here, a {\em unilateral} mechanism is a randomized mechanism whose
(random) output depends on the ballot of a single distinguished voter
$i^*$ only. Note that a unilateral truthful mechanism does not have to be a
dictatorship. For instance, the mechanism that elects with probability
$\frac{1}{2}$ each of the two top candidates according to the ballot of
voter $i^*$ is a unilateral truthful mechanism.  
A {\em duple} mechanism is an ordinal mechanism for which there are two distinguished candidates so that all other candidates are elected with probability 0, for all valuation profiles.

An optimistic interpretation of Gibbard's 1977 result as opposed to
his 1973 result that was suggested, e.g., by Barbera \cite{Barbera79}, is that
the class of randomized truthful mechanisms is quite rich and contains
many arguably ``reasonable'' mechanisms--in contrast to dictatorships, 
which are
clearly ``unreasonable''. However, we are not aware of any suggestions
in the social choice literature of any well-defined quality measures
that would enable us to rigorously compare these mechanisms and in particular
find the best. Fortunately, one of the main conceptual
contributions from computer science to mechanism design in general is
the suggestion of one such measure, namely the notion of worst case
approximation ratio relative to some objective function. Indeed, a
large part of the computer science literature on mechanism design
(with or without money) is the construction and analysis of
approximation mechanisms, following the agenda set by the seminal papers by Nisan and Ronen \cite{NisanRonen} for the case of mechanisms with money and Procaccia and Tennenholz \cite{PT} for the case of mechanisms without money (i.e., social choice functions). Following this research program, and using
Gibbard's characterization, Procaccia \cite{Procaccia10} gave in a paper conceptually 
very
closely related to the present one, upper and lower bounds on the
approximation ratio achievable by ordinal mechanisms for
various objective functions under general preferences. However, he only considered objective
functions that can be defined ordinally (such as, e.g., Borda count),
and did in particular not consider approximating the optimal social
welfare, as we do in the present paper.

The (approximate) optimization of social welfare (i.e. sum of valuations) is indeed a very
standard objective in mechanism design. In particular, in the setting
of mechanisms {\em with} money and agents with quasi-linear utilities,
the celebrated class of Vickrey-Clarke-Groves (VCG) mechanisms exactly
optimize social welfare, while classical negative results such as
Roberts' theorem, state that under general cardinal preferences (and
subject to some qualifications), weighted social welfare
is the only objective one can maximize exactly
truthfully, even with money (see Nisan \cite{Nisan07} for an exposition of all these results). It therefore seems to us extremely natural to try to
understand how well one can approximate this objective truthfully
without money under general cardinal preferences. 
One possible reason that the problem was not considered
previously to this paper (to the best of our knowledge) is that,
arguably, social welfare is a somewhat less natural objective function
without the assumption of quasi-linearity of utilities made in
the setting of mechanisms with money. Indeed, assuming quasi-linearity
essentially means forcing the valuations of all agents to be in the unit of
dollars, making it natural to subsequently add them up. On the other hand, in the
setting of social choice theory, the valuation functions are to be
interpreted as von Neumann-Morgenstern utilities (i.e, they are meant
to encode orderings on lotteries), and in particular are only
well-defined up to affine transformations. In this setting, the social
welfare has to be defined as above, as the result of adding up the valuations of all
players, {\em after} these are normalized by scaling to, say, the interval [0,1]. While this is arguably {\em ad hoc}, we note again that optimizing social welfare in this sense is in fact the intended (hoping for truthful ballots) outcome of the well known {\em range voting scheme} (
\verb=http://en.wikipedia.org/wiki/Range_voting=) which is a good piece of evidence for its naturalness.\footnote{It was pointed out to us that it is not completely clear that it is part of the range voting scheme that voters are asked to calibrate their scores so that 0 is the score of their least preferred candidate and 1 is the score of their most preferred candidate. However, without {\em some} calibration instructions, the statement "Score the candidates on a scale from 0 to 1" simply does not make sense and we believe that the present calibration instructions are the most natural ones imaginable.}

As already noted, social welfare is a cardinal objective, i.e., it depends on
the actual numerical valuations of the voters; not just their rankings
of the candidates. While it makes perfect sense to measure how well
ordinal mechanisms can approximate a cardinal objective, such as
social welfare, it certainly also makes sense to see if improvements
to the approximation of the optimal social welfare can be made by mechanisms that actually look at
the numerical scores on the ballots and not just the rankings, i.e.,
cardinal mechanisms. The limitations of ordinal mechanisms were considered recently by Boutilier {\em et al.} \cite{Boutilier} in a very interesting paper closely related to the present one, but crucially, their work did not consider incentives, i.e., they did not require truthfulness of the mechanisms in their investigations. On the other hand, truthfulness is the pivotal property in our approach.
The characterization of truthful mechanisms of Theorem \ref{thm-Gibbard77} does not apply to cardinal
mechanisms.  But noting that the definition of ``unilateral'' is not
restricted to ordinal mechanisms, one might naturally suspect that a
similar characterization would also apply to cardinal mechanisms. In a followup paper, Gibbard \cite{Gibbard78}, indeed proved a theorem along those lines,
but interestingly, his result does {\em not} apply to truthful {\em
direct} revelation mechanisms (i.e., strategy-proof social choice
functions), which is the topic of the present paper, but only to {\em
indirect} revelation mechanisms with finite strategy space. Also, the
restriction to finite strategy space (which is in direct contradiction
to direct revelation) is crucial for the proof. Somewhat
surprisingly, to this date, a characterization for the cardinal case
is still an open problem! For a discussion of this situation and for interesting counterexamples to tempting characterization attempts similar to the characterization for the ordinal case of Gibbard \cite{Gibbard77}, see 
\cite{Barbera98,Barbera10}, the bottomline being that we at the moment do not have a good understanding of what can be done with cardinal truthful mechanisms for general preferences.
Concrete examples of cardinal mechanisms for general 
preferences were given in a number of a papers in the economics and social choice
literature \cite{Zeckhauser73,Freixas84,Barbera98} and the computer
science literature \cite{Feige10}. It is interesting that while the
social choice literature gives examples suggesting that the space of
cardinal mechanisms is rich and even examples of instances where a cardinal
mechanism for voting can yield a (Pareto) better result than all
ordinal mechanisms \cite{Freixas84}, there was apparently no systematic
investigation into constructing ``good'' cardinal mechanisms for
unrestricted preferences. Here, as in the ordinal case, we suggest that the notion of
approximation ratio provides a meaningful measure of quality that
makes such investigations possible, and indeed, our present paper is meant to start such investigations. Our investigations are very much helped by
the work of Feige and Tennenholtz \cite{Feige10} who considered and characterized the {\em
strongly truthful, continuous, unilateral} cardinal mechanisms.
While their agenda was mechanisms for which the objective is information
elicitation itself rather than mechanisms for approximate optimization
of an objective function, the mechanisms they suggest still turn out to be useful for social welfare optimization. In particular, our construction establishing the gap between the approximation ratios for cardinal and ordinal mechanisms for three candidates is based on their {\em quadratic lottery}.

\subsection{Organization of paper}
In Section \ref{sec-prel} we give formal definitions of the concepts
informally discussed above, and state and prove some useful lemmas. In
Section \ref{sec-many}, we present our results for an
arbitrary number of candidates $m$. In Section \ref{sec-few}, we present our results for $m=3$. We conclude with a discussion of open problems in Section \ref{sec-conc}.

\section{Preliminaries}\label{sec-prel}
We let $\valset{m}$ denote the set of canonically represented valuation functions on $M = \{1,2,\ldots,m\}$. That is, $\valset{m}$ is
  the set of injective functions $u: M \rightarrow [0,1]$
  with the property that $0$ as well as $1$ are contained in the image
  of $u$.  

We let $\mech{m}{n}$ denote the set of truthful
mechanisms for $n$ voters and $m$ candidates. That is, $\mech{m}{n}$
is the set of random maps $J: {\valset{m}}^n \rightarrow M$ with the property that for voter $i \in \{1,\ldots,n\}$, and all ${\bf u} = (u_i, u_{-i}) \in
{\valset{m}}^n$ and $\tilde u_i \in \valset{m}$, we have 
$E[u_i(J(u_i,u_{-i})] \geq E[u_i(J(\tilde u_i, u_{-i})]$. Alternatively, instead of viewing a mechanism as a random map, we can view it as a map from ${\valset{m}}^n$ to $\Delta_m$, the set of probability density functions on $\{1,\ldots,m\}$. With this interpretation, note that $\mech{m}{n}$ is a convex subset of the vector space of all maps from ${\valset{m}}^n$ to ${\mathbb{R}}^m$.

We shall be interested in certain special classes of
mechanisms. In the following definitions, we throughout view a mechanism $J$ as
a map from
${\valset{m}}^n$ to $\Delta_m$.

An {\em ordinal} mechanism $J$
is a mechanism with the following property:
$J(u_i, u_{-i}) = J(u'_i, u_{-i})$, for any voter $i$, any preference
profile ${\bf u} = (u_i, u_{-i})$, and any valuation function $u'_i$
with the property that for all pairs of candidates $j, j'$, it is the
case that $u_i(j) < u_i(j')$ if and only if $u'_i(j) <
u'_i(j')$. Informally, the behavior of an ordinal mechanism only depends on the
ranking of candidates on each ballot; not on the numerical valuations.
We let $\rmech{O}{m}{n}$ denote those mechanisms in $\mech{m}{n}$ that are ordinal. 

Following Barbera \cite{Barbera79}, we define an {\em anonymous}
 mechanism $J$ as one that does not depend on the names of
 voters. Formally, given any permutation $\pi$ on $N$,
 and any ${\bf u} \in (\valset{m})^n$, we have 
$J({\bf u}) = J(\pi \cdot {\bf u})$, where $\pi \cdot {\bf u}$
 denotes the vector $(u_{\pi(i)})_{i=1}^n$.

Similarly following Barbera \cite{Barbera79}, we define a {\em neutral}
 mechanism $J$ as one that does not depend on the names of
 candidates. Formally, given any permutation $\sigma$ on $M$, any
 ${\bf u} \in (\valset{m})^n$, and any candidate $j$, we have $J({\bf u})_{\sigma(j)} = J(u_1 \circ \sigma, u_2 \circ \sigma, \ldots, u_n \circ \sigma)_j$.

Following \cite{Gibbard77,Barbera98}, a {\em unilateral} mechanism is a mechanism for which there
exists a single voter $i^*$ so that for all valuation profiles
$(u_{i^*}, u_{-i^*})$ and any alternative valuation profile $u'_{-i^*}$
for the voters except $i^*$, we have $J(u_{i^*}, u_{-i^*}) = J(u_{i^*}, u'_{-i^*})$.
Note that $i^*$ is {\em not} allowed to be chosen at random in the
definition of a unilateral mechanism. In this paper, we shall say that
a mechanism is {\em mixed-unilateral} if it is 
%a mechanism that
%can be viewed as first selecting a voter $i^*$ uniformly at random and
%then applying some truthful-in-expectation one-voter mechanism to the
%ballot of that voter. Note that a mixed-unilateral mechanism is
%anonymous and a special case of 
a convex combination of unilateral truthful
mechanisms.  Mixed-unilateral mechanisms are quite attractive
seen through the ``computer science lens'': They are mechanisms of
{\em low query complexity}; consulting only a single randomly chosen
voter, and therefore deserve special attention in their own right.
We let $\rmech{U}{m}{n}$ denote those mechanisms in $\mech{m}{n}$ that
are mixed-unilateral. Also, we let $\rmech{OU}{m}{n}$ denote
those mechanisms in $\mech{m}{n}$ that are ordinal as well as mixed-unilateral.
%The following characterization is not hard to show,
%further motivating the notion of mixed-unilateral.
%\begin{theorem}\label{thm-a}
%$\rmech{U}{m}{n}$ is exactly the set of anonymous mechanisms that can be realized by sampling from probability measures on the set of unilateral truthful-in-expectation mechanisms.
%\end{theorem}
%\begin{proof}
%TODO
%\end{proof}

Following Gibbard \cite{Gibbard77}, a {\em duple} mechanism
$J$ is an ordinal\footnote{Barbera {\em et al.} \cite{Barbera98} gave a much more general definition of duple mechanism; their duple mechanisms are not restricted to be ordinal. In this paper, ``duple'' refers exclusively to Gibbard's original notion.}
mechanism for which there exist two candidates $j^*_1$ and
$j^*_2$ so that for all valuation profiles, $J$ elects all other candidates with probability $0$.
%and a mechanism $J'$ for two candidates (i.e., $J'$ maps $\valset{2}^n$ to $\De%lta_{\{1,2\}}$), so that for all preference profiles
%${\bf u}$ and $k = 1,2$, we have $J({\bf u})_{j^*_k} = J'({\bf u}')_k$, 
%where for all voters $i$ and $k=1,2$, $u'_i(1) = 1$ if
%$u_i(j^*_1) > u_i(j^*_2)$ and $u'_i(1) = 0$ if $u_i(j^*_1) <
%u_i(j^*_2)$, and $u'_i(2) = 1 - u'_i(1)$. In other words, a duple
%mechanism first eliminates all candidates except two fixed ones and then runs a
%two-candidate truthful-in-expectation mechanism. Note that a duple mechanism is% necessarily ordinal, due to the fact that utilities must be rescaled to $[0,1]%$ before applying $J'$.

%Finally, we say that a mechanism is {\em regular} if it can be realized by sampling from a probability measure on the union of the set of unilateral truthful-in-expectation mechanisms and the set of duple truthful-in-expectation mechanisms. Note that Theorem \ref{thm-Gibbard77} states that any truthful-in-expectation ordinal mechanism is regular, with the extra property that the probability measure is discrete, i.e., the mechanism is a finite convex combination, and all components of this convex combination are also ordinal.
%We let $\rmech{R}{m}{n}$ denote those mechanisms in $\mech{m}{n}$ that are regular. 

We next give names to some specific important mechanisms.
We let $\unilat{m}{n}{q} \in \rmech{OU}{m}{n}$ be the mechanism for $m$ candidates and $n$ voters that picks a voter uniformly at random, and elects uniformly at random a candidate among his $q$ most preferred candidates. 
We let {\em random-favorite} be a nickname for $\unilat{m}{n}{1}$ and {\em random-candidate} be a nickname for $\unilat{m}{n}{m}$.
We let $\duple{m}{n}{q} \in \rmech{O}{m}{n}$, for $\lfloor n/2 \rfloor + 1 \leq q \leq n + 1$, be the mechanism for $m$ candidates and $n$ voters that picks two candidates
  uniformly at random and eliminates all other candidates. 
  It then checks for each voter which of the two candidates
  he prefers and gives that candidate a ``vote''. If a candidate gets at least $q$ votes, she is elected. Otherwise, a coin is flipped to
  decide which of the two candidates is elected. We let {\em random-majority} be a nickname for
  $\duple{m}{n}{\lfloor n/2 \rfloor + 1}$. Note also that $\duple{m}{n}{n+1}$ is just another name for {\em random-candidate}.
Finally, we shall be interested in the following mechanism $Q_n$ for three candidates shown to be in $\rmech{U}{3}{n}$ by Feige and Tennenholtz \cite{Feige10}: Select a voter uniformly at random, and let $\alpha$ be the valuation of his second most preferred candidate. Elect his most preferred candidate with probability $(4-\alpha^2)/6$, his second most preferred candidate with probability $(1+2 \alpha)/6$ and his least preferred candidate with probability $(1-2 \alpha + \alpha^2)/6$. We let {\em quadratic-lottery} be a nickname for $Q_n$. Note that {\em quadratic-lottery} is not ordinal. Feige and Tennenholtz \cite{Feige10} in fact presented several explicitly given non-ordinal one-voter truthful mechanisms, but {\em quadratic-lottery} is particularly amenable to an approximation ratio analysis due to the fact that the election probabilities are quadratic polynomials. 

We let $\mratio(J)$ denote the approximation ratio of a mechanism $J \in \mech{m}{n}$, when the objective is social welfare. That is,
\[ \mratio(J) = \inf_{{\bf u} \in {\valset{m}}^n} \frac{E[\sum_{i=1}^n u_i(J({\bf u}))]}{\max_{j \in M}\sum_{i=1}^n u_i(j)}. \]

We let $\ratio{m}{n}$ denote the best possible approximation ratio when there are $n$ voters and $m$ candidates. That is, 
$\ratio{m}{n} = \sup_{J \in \mbox{\small \rm Mech}_{m,n}} \mratio(J)$.
Similarly, we let $\gratio{C}{m}{n} = \sup_{J \in 
\mbox{\small \rm Mech}^{\bf C}_{m,n}} \mratio(J),$
for {\bf C} being either {\bf O}, {\bf U} or {\bf OU}. 
We let $\aratio{m}$ denote the asymptotically best possible approximation ratio when the number of voters approaches infinity. That is, $\aratio{m} = \liminf_{n \rightarrow \infty} \ratio{m}{n}$,
and we also extend this notation to the restricted classes of mechanisms with the obvious notation $\agratio{O}{m}, \agratio{U}{m}$ and $\agratio{OU}{m}$. 

The importance of neutral and anonymous mechanisms is apparent from the following simple lemma:
\begin{lemma}\label{lem-naenough}
For all $J \in \mech{m}{n}$,  there is a $J' \in \mech{m}{n}$ so that $J'$ is anonymous and neutral and so that $\mratio(J') \geq \mratio(J)$. Similarly, 
for all $J \in \rmech{C}{m}{n}$,  there is $J' \in \rmech{C}{m}{n}$ so that $J'$ is anonymous and neutral and so that $\mratio(J') \geq \mratio(J)$, for {\bf C} being either {\bf O}, {\bf U} or {\bf OU}.
\end{lemma}
\begin{proof}
Given any mechanism $J$, we can ``anonymize'' and ``neutralize'' $J$
by applying a uniformly chosen random permutation to the set of
candidates and an independent uniformly chosen random permutation to
the set of voters before applying $J$. This yields an anonymous and
neutral mechanism $J'$ with at least a good an approximation ratio as
$J$. Also, if $J$ is ordinal and/or mixed-unilateral, then so is $J'$.
\end{proof}

Lemma \ref{lem-naenough} makes the characterizations of the following theorem very useful.
\begin{theorem}\label{thm-na}
The set of anonymous and neutral mechanisms in $\rmech{OU}{m}{n}$ is equal to the set of convex combinations of the mechanisms
$\unilat{m}{n}{q}$, for $q \in \{1,\ldots,m\}$. 
%It is also exactly the set of anonymous and neutral mechanisms in $\mech{m}{n}$ obtainable by sampling from probability measures on the set of ordinal truthful-in-expectation unilateral mechanisms.
Also, the set of anonymous and neutral mechanisms in $\mech{m}{n}$ that can be obtained as convex combinations of duple mechanisms is equal to the set of convex combinations of the mechanisms 
$\duple{m}{n}{q}$, for $q \in
\{\lfloor n/2 \rfloor + 1, \lfloor n/2 \rfloor + 2, \ldots,
n, n+1\}$.
\end{theorem}
\begin{proof}
A very closely related statement was shown by Barbera \cite{Barbera78}. We sketch how to derive the theorem from that statement.

Barbera (in \cite{Barbera78}, as summarized in the proof of Theorem 1
in \cite{Barbera79}) showed that the anonymous, neutral mechanisms in
$\rmech{OU}{m}{n}$ are exactly the {\em point voting schemes}
and that the anonymous, neutral mechanism that are convex combinations of duple mechanism are exactly {\em
supporting size schemes}.
A point voting scheme is given by $m$ real
numbers $(a_j)_{j=1}^m$ summing to $1$, with $a_1 \geq a_2 \geq \cdots
\geq a_m \geq 0$. It picks a voter uniformly at random, and elects the
candidate he ranks $k$th with probability $a_k$, for $k=1,\ldots,m$. 
It is easy to see that
the point voting schemes are exactly the convex combinations of
$\unilat{m}{n}{q}$, for $q \in \{1,\ldots,m\}$.
A supporting size
scheme is given by $n+1$ real numbers $(b_i)_{i=0}^{n}$ with $b_n \geq
b_{n-1} \cdots \geq b_0 \geq 0$, and $b_i + b_{n-i} = 1$ for $i \leq
n/2$. It picks two different candidates $j_1, j_2$ uniformly at random
and elects candidate $j_k, k=1,2$ with probability $b_{s_k}$ where
$s_k$ is the number of voters than rank $j_k$ higher than
$j_{3-k}$. It is easy to see that the supporting size schemes are
exactly the convex combinations of $\duple{m}{n}{q}$, for $q \in
\{\lfloor n/2 \rfloor + 1, \lfloor n/2 \rfloor + 2, \ldots, n + 1\}$.
\end{proof}

The following corollary is immediate from Theorem \ref{thm-Gibbard77} and Theorem \ref{thm-na}.
\begin{corollary}\label{cor-ud} The ordinal, anonymous and neutral mechanisms in
$\mech{m}{n}$ are exactly the convex combinations of the mechanisms
$\unilat{m}{n}{q}$, for $q \in \{1,\ldots,m\}$ and $\duple{m}{n}{q}$,
for $q \in \{\lfloor n/2 \rfloor + 1, \lfloor n/2 \rfloor + 2, \ldots,
n\}$.
\end{corollary}

We next present some lemmas that allow us to understand the asymptotic behavior of $\ratio{m}{n}$ and $\gratio{C}{m}{n}$ for fixed $m$ and large $n$, for {\bf C} being either {\bf O}, {\bf U} or {\bf OU}.

\begin{lemma}\label{lem-scale}
For any positive integers $n,m,k$, we have
$\ratio{m}{kn} \leq
\ratio{m}{n}$ and $\gratio{C}{m}{kn} \leq \gratio{C}{m}{n}$, for {\bf C} being either {\bf O}, {\bf U} or {\bf OU}. 
\end{lemma} 

\begin{proof}
Suppose
we are given any mechanism $J$ in $\mech{m}{kn}$ with approximation ratio
$\alpha$. We will convert it to a mechanism $J'$ in $\mech{m}{n}$
with the same approximation ratio, hence proving $\ratio{m}{kn} \leq
\ratio{m}{n}$. The natural idea is to let $J'$ simulate $J$ on the profile where we simply make $k$ copies of each of the $n$ ballots. More specifically, let $\mathbf{u'}=\left(u'_1,\ldots,u'_n\right)$ be a valuation profile with $n$ voters and $\mathbf{u}=\left(u_1,\ldots,u_{kn}\right)$ be a valuation profile with $kn$ voters, such that $u_{ik+1}=u_{ik+2}=\ldots=u_{(i+1)k}=u'_{i
+1}$, for $i=0,\ldots,n-1$, where ``$=$'' denotes component-wise equality. Then let $J'(\mathbf{u'})=J(\mathbf{u})$. To complete the proof, we need to prove that if $J$ is truthful, $J'$ is truthful as well.

Let $\mathbf{u}=\left(u_1,\ldots,u_{kn}\right)$ be the profile defined above for $kn$ agents and let $\mathbf{u'}$ be the corresponding $n$ agent profile. We will consider deviations of agents with the same valuation functions to the same misreported valuation vector $\hat{u}$; without loss of generality, we can assume that these are agents $1,\ldots,k$. For ease of notation, let $\mathbf{u_{i+1}}=(u_{ik+1}=u_{ik+2}=\ldots=u_{(i+1)k})$ be a \emph{block} of valuation functions, for $i=0,\ldots,n-1$ and note that given this notation, we can write $\mathbf{u}=\left(\mathbf{u_1},\mathbf{u_2},\ldots,\mathbf{u_n}\right)=\left(u_1,\ldots,u_k,\mathbf{u_2},\ldots,\mathbf{u_n}\right)$. Let $v^{*}=E[u_i(J(\mathbf{u}))]$.  Now consider the profile $\left(\hat{u},u_2,\ldots,u_k,\mathbf{u_2},\ldots,\mathbf{u_n}\right)$. By truthfulness, it holds that agent $1$'s expected utility in the new profile (and with respect to $u_1$) is at most $v^{*}$. Next, consider the profile $\left(\hat{u},\hat{u},u_3,\ldots,u_k,\mathbf{u_2},\ldots,\mathbf{u_n}\right)$ and observe that agent $2$'s utility from misreporting should be at most her utility before misreporting, which is at most $v^{*}$. Continuing like this, we obtain the valuation profile $\left(\hat{u},\hat{u},\ldots,\hat{u},\mathbf{u_2},\ldots,\mathbf{u_n}\right)$ in which the expected utility of agents $1,\ldots,k$ is at most $v^{*}$ and hence no deviating agent gains from misreporting. Now observe that the new profile $\left(\hat{u},\hat{u},\ldots,\hat{u},\mathbf{u_2},\ldots,\mathbf{u_n}\right)$ corresponds to an $n$-agent profile $(\mathbf{\hat{u_i}',u_{-i}'})=\left(\hat{u}_1',u'_2,\ldots,u'_n\right)$ which is obtained from $\mathbf{u'}$ by a single miresport of agent $1$. By the discussion above and the way $J'$ was constructed, agent $1$ does not benefit from this misreport and since the misreported valuation function was arbitrary, $J'$ is truthful.

The same proof works for $\gratio{C}{m}{kn} \leq \gratio{C}{m}{n}$, for {\bf C} being either {\bf O}, {\bf U} or {\bf OU}.     
\end{proof}

\begin{lemma}\label{lem-kill}
For any $n,m$ and $k < n$, we have
$\ratio{m}{n} \geq \ratio{m}{n-k} - \frac{km}{n}.$ 
Also, 
$\gratio{C}{m}{n} \geq \gratio{C}{m}{n-k} - \frac{km}{n},$
for {\bf C} being either {\bf O}, {\bf U}, or {\bf OU}.
\end{lemma}

\begin{proof}
We construct a mechanism $J'$ in $\mech{m}{n}$ from a mechanism $J$ in
$\mech{m}{n-k}$. The mechanism $J'$ simply simulates $J$ after removing
$k$ voters, chosen uniformly at random and randomly mapping the remaining voters to $\{1,\ldots,n\}$.
In particular, if $J$ is ordinal (or mixed-unilateral, or both)
then so is $J'$. Suppose $J$ has approximation ratio
$\alpha$. Consider running $J'$ on any profile where the socially
optimal candidate has social welfare $w^*$. Note that $w^* \geq n/m$, since each voter assigns valuation $1$ to some candidate. Ignoring $k$ voters
reduces the social welfare of any candidate by at most $k$, so $J'$ is
guaranteed to return a candidate with expected social welfare at least
$\alpha (w^* - k)$. This is at least a $\alpha (1 - k/w^*) \geq \alpha - \frac{km}{n}$ fraction of $w^*$. Since the profile was arbitrary, we are done.
\end{proof}

\begin{lemma}\label{lem-limit}For any $m, n \geq 2, \epsilon>0$ and all $n' \geq (n-1)m/\epsilon$, we have $\ratio{m}{n'} \leq \ratio{m}{n} + \epsilon$ and $\gratio{C}{m}{n'} \leq \gratio{C}{m}{n} + \epsilon$,
for {\bf C} being either {\bf O}, {\bf U}, or {\bf OU}.
\end{lemma}

\begin{proof}
If $n$ divides $n'$, the statement follows from Lemma \ref{lem-scale}. Otherwise, let $n^*$ be the smallest number larger than $n'$ divisible by $m$; we have $n^* < n' + n$. By Lemma \ref{lem-scale}, we have 
$\ratio{m}{n^*} = \ratio{m}{n}$. By Lemma \ref{lem-kill}, we have
$\ratio{m}{n^*} \geq \ratio{m}{n'} - \frac{(n-1)m}{n^*}$.
Therefore, $\ratio{m}{n'} \leq \ratio{m}{n} + \frac{(n-1)m}{n^*} \leq \ratio{m}{n} + \frac{(n-1)m}{n'}$. The same arguments work for proving $\gratio{C}{m}{n'} \leq \gratio{C}{m}{n} + \epsilon$,
for {\bf C} being either {\bf O}, {\bf U}, or {\bf OU}. 
\end{proof}

In particular, Lemma \ref{lem-limit} implies that $\ratio{m}{n}$ converges to a limit as $n \rightarrow \infty$.

\subsection{Quasi-combinatorial valuation profiles}

It will sometimes be useful to restrict the set of valuation functions to a certain finite domain $\rvalset{m}{k}$ for an integer parameter $k \geq m$. Specifically, we define: 
\[ \rvalset{m}{k} = \left\{ u \in \valset{m} | \Im(u) \subseteq \{0, \frac{1}{k}, \frac{2}{k}, \ldots, \frac{k-1}{k}, 1\}\right\} \]
where $\Im(u)$ denotes the image of $u$.
Given a valuation function $u \in \rvalset{m}{k}$, we define its {\em alternation number} $\altern(u)$ as
\[ a(u) = \# \{j \in \{0,\ldots,k-1\} | [\frac{j}{k} \in \Im(u)] \oplus [\frac{j+1}{k}  \in \Im(u)] \}, \]
where $\oplus$ denotes exclusive-or. That is, the alternation number of $u$ is the number of indices $j$ for which exactly one of $j/k$ and $(j+1)/k$ is in the image of $u$. Since $k \geq m$ and $\{0,1\} \subseteq \Im(u)$, we have that the alternation number of $u$ is at least $2$. We shall be interested in the class of valuation functions $\combset{m}{k}$ with minimal alternation number. Specifically, we define:
\[  \combset{m}{k}  =  \{ u \in \rvalset{m}{k} |  \altern(u) = 2 \} \]
and shall refer to such valuation functions as {\em quasi-combinatorial valuation functions}. Informally, the quasi-combinatorial valuation functions have all valuations as close to 0 or 1 as possible.

The following lemma will be very useful in later sections. It states that in order to analyse the approximation ratio of an ordinal and neutral mechanism, it is sufficient to understand its performance on quasi-combinatorial valuation profiles.
\begin{lemma}\label{lem-quasicomb}
Let $J \in \mech{m}{n}$ be ordinal and neutral. Then
\[  \mratio(J) = \liminf_{k \rightarrow \infty} \min_{{\bf u} \in (\combset{m}{k})^n} \frac{E[\sum_{i=1}^n u_i(J({\bf u}))]}{\sum_{i=1}^n u_i(1)}. \]
\end{lemma} 
\begin{proof}
For a valuation profile ${\bf u} = (u_i)$, define 
$g({\bf u}) = \frac{E[\sum_{i=1}^n u_i(J({\bf u}))]}{\sum_{i=1}^n u_i(1)}.$
We show the following equations:
\begin{eqnarray}
\mratio(J) 
& = & \inf_{{\bf u} \in \valset{m}^n} \frac{E[\sum_{i=1}^n u_i(J({\bf u}))]}{\max_{j \in M}\sum_{i=1}^n u_i(j)} \\
& = & \inf_{{\bf u} \in \valset{m}^n} g({\bf u}) \label{eq-fixed-cand}\\ 
& = & \liminf_{k \rightarrow \infty} \min_{{\bf u} \in (\rvalset{m}{k})^n} g({\bf u}) \label{eq-discretized}\\
& = & \liminf_{k \rightarrow \infty} \min_{{\bf u} \in (\combset{m}{k})^n} g({\bf u}) \label{eq-quasicomb}
\end{eqnarray}

Equation (\ref{eq-fixed-cand}) follows from the fact that since $J$ is neutral, it is invariant over permutations of the set of candidates, so there is a worst case instance (with respect to approximation ratio) where the socially optimal candidate is candidate 1.
Equation (\ref{eq-discretized}) follows from the facts that
(a) each ${\bf u} \in (\valset{m})^n$ can be written as ${\bf u} = \lim_{k \rightarrow \infty} {\bf v}_k$ where $({\bf v}_k)$ is a sequence so that ${\bf v}_k \in (\rvalset{m}{k})^n$ and where the limit is with respect to the usual Euclidean topology (with the set of valuation functions being considered as a subset of a finite-dimensional Euclidean space), and
(b) the map $g$ is continuous in this topology (to see this, observe that the denominator in the formula for $g$ is bounded away from 0).
Finally, equation (\ref{eq-quasicomb}) follows from the following claim:
\[ \forall {\bf u}  \in (\rvalset{m}{k})^n\  \exists {\bf u}' \in (\combset{m}{k})^n :  g({\bf u}') \leq g({\bf u}). \]
With ${\bf u} = (u_1, \ldots, u_n)$,  we shall prove this claim by induction in $\sum_i \altern(u_i)$ (recall that $\altern(u_i)$ is the alternation number of $u_i$).

For the induction basis, the smallest possible value of $\sum_i \altern(u_i)$ is $2n$, corresponding to all $u_i$ being quasi-combinatorial. For this case, we let ${\bf u}'={\bf u}$.

For the induction step, consider a valuation profile ${\bf u}$ with $\sum_i \altern(u_i) > 2n$. Then, there must be an $i$ so that the alternation number $\altern(u_i)$ of $u_i$ is strictly larger than 2 (and therefore at least 4, since alternation numbers are easily seen to be even numbers). Then, there must be $r, s \in \{2,3,\ldots,k-2\}$, so that $r \leq s$, $\frac{r-1}{k} \not \in \Im(u_i)$, 
$\{\frac{r}{k}, \frac{r+1}{k}, \ldots, \frac{s-1}{k}, \frac{s}{k} \} \subseteq \Im(u_i)$ and $\frac{s+1}{k} \not \in \Im(u_i)$. Let $\tilde{r}$ be the largest number strictly smaller than $r$ for which $\frac{\tilde{r}}{k} \in \Im(u_i)$; this number exists since $0 \in \Im(u_i)$. Similarly, let $\tilde{s}$ be the smallest number strictly larger than $s$ for which $\frac{\tilde{s}}{k} \in \Im(u_i)$; this number exists since $1 \in \Im(u_i)$.  We now define a valuation function $u^x \in \valset{m}$ for any $x \in [\tilde{r}-r+1; \tilde{s}-s-1]$, as follows: $u^x$ agrees with $u_i$ on all candidates $j$ {\em not} in $u_i^{-1}(\{\frac{r}{k}, \frac{r+1}{k}, \ldots, \frac{s-1}{k}, \frac{s}{k} \})$, 
while for candidates 
$j \in u_i^{-1}(\{\frac{r}{k}, \frac{r+1}{k}, \ldots, \frac{s-1}{k}, \frac{s}{k} \})$
, we let 
$u^x(j) = u_i(j) + \frac{x}{k}$.  Now consider the function $h: x \rightarrow g((u^x, u_{-i}))$, where
$(u^x, u_{-i})$ denotes the result of replacing $u_i$ with $u^x$ in the profile ${\bf u}$. Since $J$ is ordinal, we see by inspection of the definition of the function $g$, that $h$ on the domain $[\tilde{r}-r+1; \tilde{s}-s-1]$ is a fractional linear function $x \rightarrow (ax + b)/(cx + d)$ for some $a,b,c,d \in {\mathbb R}$. As $h$ is defined on the entire interval $[\tilde{r}-r+1; \tilde{s}-s-1]$, we therefore have that $h$ is either monotonely decreasing or monotonely increasing in this interval, or possibly constant. If $h$ is monotonely increasing, we let $\tilde{\bf u} = (u^{\tilde{r}-r+1}, u_{-i})$, and apply the induction hypothesis on $\tilde{\bf u}$. If $h$ is monotonely decreasing, we let $\tilde {\bf u} =  (u^{\tilde{s}-s-1}, u_{-i})$, and apply the induction hypothesis on $\tilde{\bf u}$. If $h$ is constant on the interval, either choice works. This completes the proof.

\begin{figure}
\label{fig-quasicomb}
\begin{tikzpicture}
\usetikzlibrary{arrows}
\usetikzlibrary{shapes}
\tikzstyle{every node}=[draw=black, inner sep=1pt]
\draw[|->] (0,0)--(17,0);
\node [draw=white,inner sep=0pt,label=below:{$R_{m,k}$}] at (17,0) {};
\node [circle,label=below:{$0$}] at (1,0) {};
\node [circle,label=below:{$\frac{1}{10}$}] at (2.5,0) {};
\node [circle,label=below:{$\frac{2}{10}$}] at (4,0) {};
\node [circle,label=below:{$\frac{3}{10}$}] at (5.5,0) {};
\node [circle,label=below:{$\frac{4}{10}$}] at (7,0) {};
\node [circle,label=below:{$\frac{5}{10}$}] at (8.5,0) {};
\node [circle,label=below:{$\frac{6}{10}$}] at (10,0) {};
\node [circle,label=below:{$\frac{7}{10}$}] at (11.5,0) {};
\node [circle,label=below:{$\frac{8}{10}$}] at (13,0) {};
\node [circle,label=below:{$\frac{9}{10}$}] at (14.5,0) {};
\node [circle,label=below:{$1$}] at (16,0) {};
\node [draw=white,inner sep=0pt,label=below:{$\in \Im(u_i)$}] at (14,1.5) {};
\node [draw=white,inner sep=0pt,label=below:{$\notin \Im(u_i)$}] at (14,-0.7) {}; 

\draw[thick,dashed,fill=gray,opacity=0.2] (0.2,-1.5) to [out=-15, in=115] (3,0.5) to [out=-75,in=-115] (4.5,-1) to [out=75,in=120] (6,0.5) to [out=-75,in=-115] (12.9,0.5) to [out=75,in=115] (15.5,-0.5) to [out=-75,in=0] (17,-1.5);

\draw[|->] (0,-4)--(17,-4);
\node [draw=white,inner sep=0pt,label=below:{$R_{m,k}$}] at (17,-4) {};
\node [circle,label=below:{$0$}] at (1,-4) {};
\node [circle,label=below:{$\frac{1}{10}$}] at (2.5,-4) {};
\node [circle,label=below:{$\frac{2}{10}$}] at (4,-4) {};
\node [circle,label=below:{$\frac{3}{10}$}] at (5.5,-4) {};
\node [circle,label=below:{$\frac{4}{10}$}] at (7,-4) {};
\node [circle,label=below:{$\frac{5}{10}$}] at (8.5,-4) {};
\node [circle,label=below:{$\frac{6}{10}$}] at (10,-4) {};
\node [circle,label=below:{$\frac{7}{10}$}] at (11.5,-4) {};
\node [circle,label=below:{$\frac{8}{10}$}] at (13,-4) {};
\node [circle,label=below:{$\frac{9}{10}$}] at (14.5,-4) {};
\node [circle,label=below:{$1$}] at (16,-4) {};
\node [draw=white,inner sep=0pt,label=below:{$\in \Im(u_i)$}] at (14,-2.5) {};
\node [draw=white,inner sep=0pt,label=below:{$\notin \Im(u_i)$}] at (14,-4.9) {}; 

\draw[thick,dashed,fill=gray,opacity=0.2] (0.2,-5.5) to [out=-15, in=115] (3,-3.5) to [out=-75,in=-115] (4.5,-5) to [out=75,in=120] (9,-4) to [out=-70,in=120] (17,-5.7);

\end{tikzpicture}
\caption{Example of the induction step of the proof of Lemma \ref{lem-quasicomb} for $m=7$ and $k=10$. Here, $r=4$, $s=7$, $\tilde{r}=2$ and $\tilde{s}=10$ and hence $x \in \left[-1,2\right]$. The bottom figure depicts the induced profile when $h(x)$ is monotonely decreasing in $[-1,2]$.}
\end{figure}
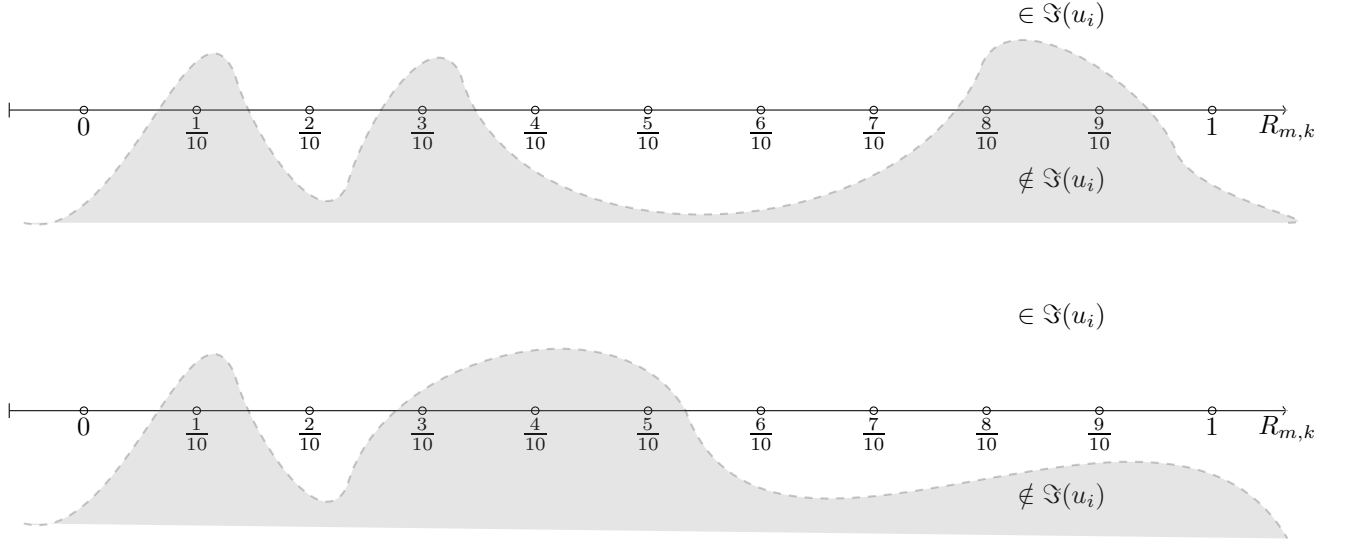

\end{proof}

\section{Mechanisms and negative results for the case of many candidates}\label{sec-many}

We can now analyze the approximation ratio of the mechanism $J \in \rmech{OU}{m}{n}$ that with probability
$3/4$ elects a uniformly random candidate and with probability
$1/4$ uniformly at random picks a voter and elects a candidate uniformly
at random from the set of his $\lfloor{m^{1/2}}\rfloor$ most
preferred candidates.  
\begin{theorem} 
Let $n \geq 2, m \geq 3$.
Let $J = \frac{3}{4} \unilat{m}{n}{m} + \frac{1}{4}
\unilat{m}{n}{\lfloor {m^{1/2}}\rfloor}$. Then, 
$\mratio(J) \geq 0.37 m^{-3/4}$.
\end{theorem}
\begin{proof}
For a valuation profile ${\bf u} = (u_i)$, we define 
$g({\bf u}) = \frac{E[\sum_{i=1}^n u_i(J({\bf u}))]}{\sum_{i=1}^n u_i(1)}.$
By Lemma \ref{lem-quasicomb}, since $J$ is ordinal, it is enough to bound from below $g({\bf u})$ for all ${\bf u} \in (\combset{m}{k})^n$ with $k \geq 1000(n m)^2$. Let $\epsilon = 1/k$. Let $\delta = m \epsilon$. Note that all functions of ${\bf u}$ map each alternative either to a valuation smaller than $\delta$ or a valuation larger than $1-\delta$. 

%Recall that $g({\bf u}) = \frac{E[\sum_{i=1}^n u_i(J({\bf u}))]}{\sum_{i=1}^n u_i(1)}$. 
Since each voter assigns valuation $1$ to at least one candidate, and since $J$ with probability $3/4$ picks a candidate uniformly at random from the set of all candidates, we have $E[\sum_{i=1}^n u_i(J({\bf u}))] \geq 3n/(4m)$. Suppose $\sum_{i=1}^n u_i(1) \leq 2 m^{-1/4} n$. Then $g({\bf u}) \geq \frac{3}{8} m^{-3/4}$, and we are done. So we shall assume from now on that 
\begin{equation}
\sum_{i=1}^n u_i(1) >  2 m^{-1/4} n. \label{eq-many-votes-for-one}
\end{equation}

Obviously, $\sum_{i=1}^n u_i(1) \leq n$. Since $J$ with probability $3/4$ picks a candidate uniformly at random from the set of all candidates, we have that $E[\sum_{i=1}^n u_i(J({\bf u}))] \geq \frac{3}{4m} \sum_{i,j} u_i(j)$. 
So if $\sum_{i,j} u_i(j) \geq \frac{1}{2} n m^{1/4}$, we have $g({\bf u}) \geq \frac{3}{8} m^{-3/4}$, and we are done. So we shall assume from now on that 
\begin{equation}
\sum_{i,j} u_i(j) <  \frac{1}{2} n m^{1/4} \label{eq-nongen}.
\end{equation}

Still looking at the fixed quasi-combinatorial ${\bf u}$, let a voter $i$ be called {\em
generous} if his $\lfloor m^{1/2} \rfloor + 1$ most preferred candidates
are all assigned valuation greater than $1 - \delta$.
Also, let a voter $i$ be called {\em friendly} if he has candidate
$1$ among his $\lfloor m^{1/2} \rfloor$ most preferred candidates.  Note that if a voter
is neither generous nor friendly, he assigns to candidate $1$ valuation at most $\delta$. This means that the total contribution
to $\sum_{i=1}^n u_i(1)$ from such voters is less than $n \delta < 0.001/m$.
Therefore, by equation (\ref{eq-many-votes-for-one}), the union of
friendly and generous voters must be a set of size at least $1.99
m^{-1/4} n$.

If we let $g$ denote the number of generous voters, we have $\sum_{i,j} u_i(j) \geq g m^{1/2} (1 - \delta) \geq 0.999 g m^{1/2}$, so by equation (\ref{eq-nongen}), we have that $0.999 g m^{1/2} < \frac{1}{2} n m^{1/4}$. In particular $g < 0.51 m^{-1/4} n$. So since the union of friendly and generous voters must be a set of size at least a $1.99 m^{-1/4} n$ voters, we conclude that there are at least $1.48 m^{-1/4} n$ friendly voters, i.e. the friendly voters is at least a $1.48 m^{-1/4}$ fraction of the set of all voters. 
But this ensures that $\unilat{m}{n}{\lfloor {m^{1/2}}\rfloor}$ elects candidate $1$ with probability at least $1.48 m^{-1/4}/m^{1/2} \geq 1.48 m^{-3/4}$. Then, $J$ elects candidate $1$ with probability at least $0.37 m^{-3/4}$ which means that $g({\bf u}) \geq 0.37 m^{-3/4}$, as desired. This completes the proof. 
\end{proof}

We next show our negative result. We show that any convex combination of (not necessarily ordinal) unilateral and duple mechanisms performs poorly. 
\begin{theorem}\label{thm-neg}
Let $m \geq 20$ and let $n = m-1+g$ where $g = \lfloor m^{2/3} \rfloor$. For any mechanism $J$ that is a convex combination of unilateral and duple mechanisms in $\mech{m}{n}$, we have $\mratio(J) \leq 5 m^{-2/3}$. 
\end{theorem}
\begin{proof}
Let $k= \lfloor m^{1/3} \rfloor$. By applying the same proof technique as in the proof of Lemma \ref{lem-naenough}, we can assume that $J$ can be decomposed into a convex combination of mechanisms $J_\ell$, with each $J_\ell$ being anonymous as well as neutral, and each $J_\ell$ either being a mechanism of the form $\duple{m}{n}{q}$ for some $q$ (by Theorem \ref{thm-na}), or a mechanism that applies a truthful one-voter neutral mechanism $U$ to a voter chosen unformly at random.

We now describe a single profile for which any such mechanism $J_\ell$ performs 
badly. Let $M_1,..,M_g$ be a partition of $\{1,\ldots,kg\}$ with $k$ candidates in each set.
The bad profile has the following voters:
\begin{itemize}
\item{}For each $i \in \{1,\ldots,m-1\}$ a voter that assigns $1$ to
candidate $i$, $0$ to candidate $m$ and valuations smaller than
$1/m^2$ to the rest.
\item{}For each $j \in \{1,\ldots,g\}$ a vote that assigns valuations strictly bigger than $1-1/m^2$ to members of $M_j$, valuation $1-1/m^2$ to m, and valuations smaller than $1/m^2$ to the rest.
\end{itemize}
Note that the social welfare of candidate $m$ is $(1-1/m^2) g$ while the
social welfares of the other candidates are all smaller than $2+{1/m}$. Thus,
the conditional expected approximation ratio given that the mechanism
does not elect $m$ is at most $(2+{1/m})/(1-1/m^2)g \leq 3 m^{-2/3}$.  
We therefore only need to estimate the probability that candidate $m$ is elected.
For a mechanism of the form $\duple{m}{n}{q}$, candidate $m$ is chosen with probability at most $2/m$, since such a mechanism first eliminates all candidates but two and these two are chosen uniformly at random. 

For a mechanism that picks a voter uniformly at random and applies a truthful one-agent neutral mechanism $U$ to the ballot of this voter, we make the following claim: Conditioned on a particular voter $i^*$ being picked, the conditional probability that $m$ is chosen is at most $1/(r+1)$, where $r$ is the number of candidates that outranks $m$ on the ballot of voter $i$. Indeed, if candidate $m$ were chosen with conditional probability strictly bigger than $1/(r+1)$, she would be chosen with strictly higher probability than some other candidate $j^*$ who outranks $m$ on the ballot of voter $i^*$. But if so, since $U$ is neutral, voter $i$ would increase his utility by switching $j^*$ and $m$ on his ballot, as this would switch the election probabilities of $j^*$ and $m$ while leaving all other election probabilities the same. This contradicts that $U$ is truthful. Therefore, our claim is correct. This means that candidate $m$ is chosen with probability at most $1/m + (g/m)\cdot(1/k) \leq 1/m + m^{2/3}/(m(m^{1/3}-1)) \leq 2 m^{-2/3}$, since $m \geq 20$.

We conclude that on the bad profile, the expected approximation ratio of any mechanism $J_\ell$ in the decomposition is at most $3 m^{-2/3} + 2 m^{-2/3} = 5 m^{-2/3}$. Therefore, the expected approximation ratio of $J$ on the bad profile is also at most $5 m^{-2/3}$. 
\end{proof}

\begin{corollary}
For all $m$, and all sufficiently large $n$ compared to $m$, any mechanism $J$ in $\rmech{O}{m}{n} \cup \rmech{U}{m}{n}$ has approximation ratio $O(m^{-2/3})$. 
\end{corollary}
\begin{proof}
Combine Theorem \ref{thm-Gibbard77}, Lemma \ref{lem-limit} and Theorem \ref{thm-neg}.
\end{proof}

As followup work to the present paper, in a working manuscript, Lee \cite{Lee14} states a lower bound of $\Omega(m^{-2/3})$ that closes the gap between our upper and lower bounds. The mechanism achieving this bound is a convex combination of random-favorite and the mixed unilateral mechanism that uniformly at random elects one of the $m^{1/3}$ most preferred candidates of a uniformly chosen voter.  The main question that we would like to answer is how well one can do with (general) cardinal mechanisms. The next theorem provides a weak upper bound.

\begin{theorem}\label{thm-dm}
All mechanisms $J \in \mech{m}{n}$ for $m,n \geq 3$ have $\mratio(J) < 0.94$.
\end{theorem}

\begin{proof}
We will prove the theorem for mechanisms in $\mech{3}{3}$. By simply adding alternatives for which every agent has valuation almost $0$ and then applying Lemma \ref{lem-limit}, the theorem holds for any $m, n \geq 3$.
 
Assume for contradiction that there exists a mechanism $J \in
\mech{3}{3}$, with $\mratio(J) \geq 0.94$. Consider the valuation profile
${\bf u}$ with three voters $\{1,2,3\}$, three candidates $\{A,B,C\}$,
and valuations $u_1(B) = u_2(B) = u_3(C) = 1$, $u_1(C) = u_2(C) =
u_3(B) = 0$, $u_1(A) = 0.7$ and $u_2(A) = u_3(A) = 0.8$.  The social
optimum on profile ${\bf u}$ is candidate $A$, with social welfare
$w_A=2.3$, while $w_B=2$ and $w_C=1$. Since $J$'s expected social
welfare is at least a $0.94$ fraction of $w_A$, i.e. $2.162$, the
probability of $A$ being elected is at least $0.54$, as otherwise the
expected social welfare would be smaller than $0.54 \cdot 2.3 + 0.46 \cdot
2 = 2.162$ . The expected utility $\tilde u$ of voter 1 in that case is
at most $0.54 \cdot 0.7+ 0.46 \cdot 1 = 0.838$.

Next, consider the profile ${\bf u}'$ identical to ${\bf u}$ except
that $u'_1(A) = 0.0001$.  Let $p_A,p_B,p_C$ be the probabilities of
candidates $A,B$ and $C$ being elected on this profile,
respectively. The social optimum is $B$ with social welfare $2$. It
must be that $0.7 p_A+ p_B \leq \tilde u$, otherwise on profile ${\bf u}$,
voter 1 would have an incentive to misreport $u_1(A)$ as $0.0001$. Also,
since $J$ has an approximation ratio of at least $0.94$, it must be
the case that $1.6001 \cdot p_A + 2 \cdot p_B + p_C \geq 1.88$. By
those two inequalities, we have: $0.9001p_A+p_B+p_C \geq 1.8800 -
\tilde u$ $\Rightarrow$ $0.9001(p_A+p_B+p_C)+0.0999p_B+0.0999p_C \geq
1.8800 - \tilde u$ $\Rightarrow$ $0.0999(p_B+p_C) \geq 0.10419$
$\Rightarrow$ $p_B+ p_C \geq 1.42$ which is not possible. Hence, it
cannot be that $\mratio(J) \geq 0.94$.
\end{proof} 

Recall that in the definition of valuation functions $u_i$, we required $u_i$ to be injective, i.e. ties are not allowed in the image of the function. If we actually allow $u_i$ to map the same real number to different candidates (with $0$ and $1$ still in the image of the function), we can prove a much stronger upper bound on the approximation ratio of any truthful mechanism. The proof is based on a bound proven by Filos-Ratsikas et al. \cite{Aris14} for the \emph{one-sided matching problem}. There are some interesting technical difficulties in adapting their proof to work for thel setting without ties. As we do not want to declare either the ``ties" or the ``no ties" model the ``right one", we want all positive results (mechanisms) to be proven for the setting with ties and all negative ones (upper bounds on approximation ratio) to be proven for the setting without ties. The proof of the following theorem is the only one of the paper which isn't easily modified to work for both settings. 

\begin{theorem}\label{thm:anyupper}
Let $J'$ be any voting mechanism for $n$ agents and $m$ alternatives, with $m \geq n^{\lfloor\sqrt{n}\rfloor +2}$, in the setting with ties. The approximation ratio of $J'$ is $O(\log\log m/\log m)$. 
\end{theorem}

\begin{proof}
Filos-Ratsikas et. al \cite{Aris14} proved a related upper bound for the one-sided matching problem.\footnote{Their proof is actually for the setting of $n$ agents and $n$ items but it can be easily adapted to work when the number of items is $\lfloor \sqrt{n} \rfloor+2$.} The bound corresponds to an upper bound on the approximation ratio of any truthful mecahnism $J$ in the general setting with ties. This is because there is a reduction from the general setting with ties to the setting of the one-sided matching problem. 

In the one-sided matching problem, there is a set of $n$ agents and a set of $k$ items and each agent $i$ has a valuation function $v_i:[k] \rightarrow [0,1]$ mapping items to real values in the unit interval. Similarly to our definitions, these functions are injective and both $0$ and $1$ are in their image. A mechanism $J$ on input a valuation profile $\mathbf{v}=(v_1,...,v_n)$ outputs a \emph{matching} $J(\mathbf{v})$, i.e. an allocation of items to agents such that each agent receives at most one item. Let $J(\mathbf{v})_i$ be the item allocated to agent $i$. For convenience, we will refer to this problem as the \emph{matching setting} and to our problem as the \emph{general setting}. 

The reduction works as follows. Let $\mathbf{v}=(v_1,...,v_n)$ be a valuation profile of the matching setting. We will construct a valuation profile $\mathbf{u}=(u_1,...,u_n)$ of the general setting that will correspond to $\mathbf{v}$. Let each outcome of the matching setting correspond to a candidate in the general setting. For every agent $i$ and every item $j$ let $u_i(A)=v_i(j)$ for each candidate $A \in M$ that corresponds to a matching in which item $j$ is allocated to agent $i$. Note that the number of candidates is $n^k$ and a bound for the matching setting implies a bound for the general setting. Specifically, the $O(1/\sqrt{n})$ bound proved in \cite{Aris14} translates to a $O(\log\log m/\log m)$ upper bound. 
\end{proof}

\section{Mechanisms and negative results for the case of three candidates}\label{sec-few}
%This was moved to previous section:
%The truthful-in-expectation, anonymous, neutral, mechanisms when the domain
%of valuations are restricted to the finite set $\rvalset{m}{k}$ is exactly the convex 
%combinations of mechanisms of the
%  form $\unilat{m}{n}{k}$ and $\duple{m}{n}{\alpha}$. This essentially
%  follows from work of Barbera \cite{Barbera79}.
 
%We already said all of the belov in the introduction, so I am outcommenting it for now.
%
%When $m=3$, trivial approximation guarantees are given by the ``simple'' rules such as random-favorite, random-candidate or random-majority. It is not difficult to see that both random-favorite and random-majority achieve an $1/2$ fraction of the optimal social welfare while random-candidate can only guarantee a $1/3$ approximation ratio. As we will see in the following, convex combinations of those ``simple'' mechanisms do beat the $1/2$ bound; in fact any anonymous, neutral strategyproof ordinal mechanism can guarantee and approximation ratio of $0.616$. With cardinal mechanisms, we get even better approximation guarantees; the quadratic lottery achieves an $0.618$ fraction of the optimal social welfare, whereas a convex combination of the quadratic lottery and random-favorite, achieves an approximation ratio of $0.666$, clearly outpeforming the best ordinal strategyproof mechanism. We begin the discussion with the case of a few voters.

In this section, we consider the special case of three candidates $m=3$. To improve readability, we shall denote the three candidates by $A, B$ and $C$, rather than by 1,2 and 3.

When the number of candidates $m$ as well as the number of voters $n$ are small constants, the exact values of $\gratio{O}{m}{n}$ and $\gratio{OU}{m}{n}$ can be determined. We first give a clean example, and then describe a general method.

%For $n=2,3,4$ or $5$, we can exactly compute $\gratio{O}{m}{n}$ and $\gratio{OU}{m}{n}$, using linear programming and a matrix game argument, as we will describe shortly.

%First, we prove the following interesting result for $3$ voters.

\begin{proposition}\label{prop-three}
For all $J \in \rmech{O}{3}{3}$, we have $\mratio(J) \leq 2/3$. 
\end{proposition} 
\begin{proof}
By Lemma \ref{lem-naenough}, we can assume that $J$ is anonymous and neutral.
Let $A >_i B$ denote the fact that voter $i$ ranks candidate $A$ higher than $B$ in his ballot. Let a {\em Condorcet} profile be any valuation profile with $A>_{1} B>_{1} C$, $B >_{2} C >_{2} A$ and $C >_{3} A >_{3} B$. Since $J$ is neutral and anonymous, by symmetry, $J$ elects each candidate with probability $1/3$. Now, for some small $\epsilon > 0$, consider the Condorcet profile where $u_1(B)= \epsilon$, $u_2(C)=\epsilon$ and $u_3(A)=1-\epsilon$. The socially optimal choice is candidate $A$ with social welfare $2-\epsilon$, while the other candidates have social welfare $1+\epsilon$. Since each candidate elected with probability $1/3$, the expected social welfare is $(4+\epsilon)/3$. By making $\epsilon$ suffciently small, the approximation ratio on the profile is arbitrarily close to $2/3$.
\end{proof}
With a case analysis and some pain, it can be proved by hand that
\emph{random-majority} achieves an approximation ratio of at least
$2/3$ on any profile with three voters and three candidates.
Together with Proposition \ref{prop-three}, this implies that
$\gratio{O}{3}{3} = \frac{2}{3}$.  Rather than presenting the case
analysis, we describe a general method for how to exactly and 
mechanically compute
$\gratio{O}{m}{n}$ and $\gratio{OU}{m}{n}$ and the associated optimal
mechanisms for small values of $m$ and $n$. The key is to apply {\em
Yao's principle} \cite{Yao77} and view the construction of a
randomized mechanism as devising a strategy for Player I in a
two-player zero-sum game $G$ played between Player I, the mechanism
designer, who picks a mechanism $J$ and Player II, the adversary, who
picks an input profile ${\bf u}$ for the mechanism, i.e., an element of
$(\valset{m})^n$. The payoff to Player I is the approximation ratio of
$J$ on ${\bf u}$. Then, the value of $G$ is exactly the
approximation ratio of the best possible randomized mechanism. In order
to apply the principle, the computation of the value of $G$ has to be
tractable. In our case, Theorem \ref{thm-na} allows us to reduce the
strategy set of Player I to be finite while Lemma \ref{lem-quasicomb}
allows us to reduce the strategy set of Player II to be finite. This
makes the game into a matrix game, which can be solved to optimality
using linear programming.  The details follow.

For fixed $m, n$ and $k > 2m$, recall that the set of quasi-combinatorial valuation functions $\combset{m}{k}$ 
is the set of valuation functions $u$ for which there is a $j$ so that
$\Im(u) = \{0, \frac{1}{k}, \frac{2}{k}, \ldots, \frac{m-j-1}{k} \} \cup \{\frac{k-j+1}{k}, \frac{k-j+2}{k}, \ldots,  \frac{k-1}{k}, 1\}$.
Note that a quasi-combinatorial valuation function $u$ is fully
described by the value of $k$, together with a partition of $M$ into
two sets $M_0$ and $M_1$, with $M_0$ being those candidates close to $0$ and
$M_1$ being those sets close to $1$ together with a ranking of the candidates (i.e., a total ordering $<$
on $M$), so that all elements of $M_1$ are greater than all elements of $M_0$ in this ordering. Let the {\em type} of a quasi-combinatorial valuation function
be the partition and the total ordering $(M_0, M_1, <)$. Then, a
quasi-combinatorial valuation function is given by its type and the
value of $k$. For instance, if $m=3$, one possible type is
$(\{B\},\{A,C\}, \{B < A < C\})$, and the quasi-combinatorial valuation
function $u$ corresponding to this type for $k=1000$ is $u(A) =
0.999$, $u(B)=0$, $u(C) = 1$. We see that for any fixed value of $m$,
there is a finite set $T_m$ of possible types. In particular, we have
$|T_3| = 12$.  Let $\eta: T_m \times {\mathbb N} \rightarrow
\combset{m}{k}$ be the map that maps a type and an integer $k$ into
the corresponding quasi-combinatorial valuation function.

For fixed $m, n$, consider the following matrices $G$ and $H$. The
matrix $G$ has a row for each of the mechanisms $\unilat{m}{n}{q}$ for
$q = 1,\ldots,m$, while the matrix $H$ has a row for each of the
mechanisms $\unilat{m}{n}{q}$ for $q = 1,\ldots,m$ as well as for each of the mechanisms
$\duple{m}{n}{q}$, for $q = \lfloor n/2 \rfloor + 1, \lfloor
  n/2 \rfloor + 2, \ldots, n$. Both matrices have a column for each element of $(T_m)^n$. The entries of the matrices are as follows: Each entry is indexed by a mechanism $J \in \mech{m}{n}$ (the row index) and by a type profile ${\bf t} \in (T_m)^n$ (the column index). We let that entry be 
\[ c_{J,{\bf t}} = \lim_{k \rightarrow \infty} \frac{E[\sum_{i=1}^n u_i(J({\bf u}^k))]}{\max_{j \in M}\sum_{i=1}^n u^k_i(j)}, \]
where $u^k_i = \eta(t_i, k)$.
Informally, we let the entry be the approximation ratio of the mechanism on the quasi-combinatorial profile of the type profile indicated in the column and with $1/k$ being ``infinitisimally small''. Note that for the mechanisms at hand, despite the fact that the entries are defined as a limit, it is straightforward to compute the entries symbolically, and they are rational numbers. 

We now have
\begin{lemma}\label{lem-matrix}
The value of $G$, viewed as a matrix game with the row player being
the maximizer, is equal to $\gratio{OU}{m}{n}$.  The value of $H$ is
equal to $\gratio{O}{m}{n}$. 
Also, the optimal strategies for the row players in
the two matrices, viewed as convex combinations of the mechanisms
corresponding to the rows, achieve those ratios.
\end{lemma}
\begin{proof}
We only show the statement for $\gratio{O}{m}{n}$, the other proof
being analogous.  For fixed $k$, consider the matrix $H^k$ defined
similarly to $H$, but with entries $c_{J,{\bf t}} = \frac{E[\sum_{i=1}^n
u_i(J({\bf u}^k))]}{\max_{j \in M}\sum_{i=1}^n u^k_i(j)}$, where
$u^k_i = \eta(t_i, k)$. Viewing $H^k$ as a matrix game, a mixed
strategy of the row player can be interpreted as a convex combination
of the mechanisms corresponding to the rows, and the expected payoff
when the column player responds with a particular column ${\bf t}$ is equal to the
approximation ratio of $J$ on the valuation profile
$(\eta(t_i,k))_i$. Therefore, the value of the game is the worst case
approximation ratio of the best convex combination, among profiles of
the form $(\eta(t_i,k))_i$ for a type profile ${\bf t}$.  By Lemma
\ref{lem-naenough}, $\gratio{O}{m}{n}$ is determined by the best available
anonymous and neutral ordinal mechanism. By Corollary \ref{cor-ud},
the anonymous and neutral ordinal mechanisms are exactly the convex
combinations of the $\unilat{m}{n}{q}$ and the $\duple{m}{n}{q}$
mechanisms for various $q$. Given any particular convex combination
yielding a mechanism $K$, by Lemma \ref{lem-quasicomb}, its worst case
approximation ratio is given by $\liminf_{k \rightarrow \infty}
\min_{{\bf u} \in (\combset{m}{k})^n} \frac{E[\sum_{i=1}^n u_i(K({\bf
u}))]}{\sum_{i=1}^n u_i(A)}$ which is equal to $\liminf_{k \rightarrow
\infty} \min_{{\bf u} \in (\combset{m}{k})^n} \frac{E[\sum_{i=1}^n
u_i(K({\bf u}^k))]}{\max_{j \in M}\sum_{i=1}^n u^k_i(j)}$, since $K$
is neutral. This means that no mechanism can have an approximation ratio better than the limit of the values of the games $H^k$ as $k$ approaches
infinity. By continuity of the value of a matrix game as a function of its entries, this is equal to the value of $H$. Therefore, $\gratio{O}{m}{n}$ is at most the value of $H$.
Now consider the mechanism $J$ defined by the optimal strategy for the row player in the
matrix game $H$. As the entries of $H_k$
converge to the entries of $H$ as $k \rightarrow \infty$, we have that
for any $\epsilon > 0$, and sufficiently large $k$, the strategy is
also an $\epsilon$-optimal strategy for $H_k$. Since $\epsilon$ is
arbitrary, we have that $\mratio(J)$ is at least the value of $H$, completing the proof.
\end{proof}
When
applying Lemma \ref{lem-matrix} for concrete values of $m,n$, one can take
advantage of the fact that all mechanisms corresponding to rows are
anonymous and neutral. This means that two different columns will have
identical entries if they correspond to two type profiles that can be
obtained from one another by permuting voters and/or candidates. This
makes it possible to reduce the number of columns drastically. After
such reduction, we have applied the theorem to $m=3$ and $n=2,3,4$ and
$5$, computing the corresponding optimal approximation ratios and optimal mechanisms. The ratios are given in Table \ref{aratios}.%, all rounded {\em down} to 3
%digits, so that the actual numbers given are all lower bounds on the ratio obtained by the mechanisms. 
 The mechanisms achieving the ratios are shown in Table \ref{probmix} and
Table \ref{probmixuni}. These mechanisms are in general not unique. Note in particular that a different approximation-optimal mechanism than {\em random-majority} was found in $\rmech{O}{3}{3}$.

\begin{table}
\caption{Approximation ratios for $n$ voters.\label{aratios}}
\begin{center}
       \begin{tabular}{ccc}
               \textbf{$\mathbf{n}$}/\textbf{Approximation ratio}  &  $\gratio{O}{3}{n}$ & $\gratio{OU}{3}{n}$ \\ \hline
               \noalign{\smallskip}
               2 & 2/3  &  2/3 \\ 
               3 & 2/3  &  105/171 \\
               4 & 2/3  &  5/8 \\
               5 & 6407/9899  &  34/55 \\
       \end{tabular}
       \hfill
\end{center}
\end{table}

\begin{table}
\caption{Mixed-unilateral ordinal mechanisms for $n$ voters. \label{probmixuni}}\begin{center}
       \begin{tabular}{cccc}
               \textbf{$\mathbf{n}$}/\textbf{Mechanism}  &  $\unilat{3}{n}{1}$ & $\unilat{3}{n}{2}$ & $\unilat{3}{n}{3}$ \\ \hline
               \noalign{\smallskip}
               2 & 1/3  &  2/3 & 0  \\ 
               3 & 9/19  &  10/19 & 0  \\ 
               4 & 1/2  &  1/2 & 0  \\ 
               5 & 5/11  &  6/11 & 0  \\  
       \end{tabular}
\end{center}
\end{table}

\begin{table}
\caption{Ordinal mechanisms for $n$ voters. \label{probmix}}
\begin{center}
       \begin{tabular}{ccccccc}
               \textbf{$\mathbf{n}$}/\textbf{Mechanism}  &  $\unilat{3}{n}{1}$ & $\unilat{3}{n}{2}$ & $\unilat{3}{n}{3}$ & $\duple{n}{3}{\lfloor n/2\rfloor +1}$ & $\duple{n}{3}{\lfloor n/2\rfloor +2}$ & $\duple{n}{3}{\lfloor n/2\rfloor +3}$ \\ \hline
               \noalign{\smallskip}
               2 & 4/100  &  8/100 & 0 & 88/100 & | & | \\ 
               3 & 47/100  &  0 & 0 & 53/100 & 0 & | \\ 
               4 & 0  &  0 & 0 & 1 & 0 & | \\ 
               5 & 3035/9899  &  0 & 0 & 3552/9899 & 3312/9899 & 0 \\ 
       \end{tabular}
\end{center}
\end{table}

We now turn our attention to the case of three candidates and
arbitrarily many voters. In particular, we shall be interested in
$\agratio{O}{3} = \liminf_{n \rightarrow \infty} \gratio{O}{3}{n}$ and
$\agratio{OU}{3} = \liminf_{n \rightarrow \infty}
\gratio{OU}{3}{n}$. By Lemma \ref{lem-limit}, we in fact have
$\agratio{O}{3} = \lim_{n \rightarrow \infty} \gratio{O}{3}{n}$ and
$\agratio{OU}{3} = \lim_{n \rightarrow \infty} \gratio{OU}{3}{n}$.

We present a family of ordinal and mixed-unilateral mechanisms $J_n$ with $\mratio(J_n) > \oul$. In particular, $\agratio{OU}{3} > \oul$. The coefficents $\cone$ and $\ctwo$ were found by trial-and-error; we present more information about how later.
\begin{theorem}\label{thm-oul} Let $\cone = \frac{77066611}{157737759}\approx 0.489$ and $\ctwo = \frac{80671148}{157737759}\approx 0.511$. Let $J_n = \cone \cdot \unilat{m}{n}{1} + \ctwo \cdot \unilat{m}{n}{2}$. For all $n$, we have $\mratio(J_n) > \oul$. 
\end{theorem}
\begin{proof}
By Lemma \ref{lem-quasicomb}, we have that $\mratio(J_n) = \liminf_{k
\rightarrow \infty} \min_{{\bf u} \in (\combset{3}{k})^n}
\frac{E[\sum_{i=1}^n u_i(J_n({\bf u}))]}{\sum_{i=1}^n u_i(A)}.$ Recall
the definition of the set of {\em types} $T_3$ of quasi-combinatorial
valuation functions on three candidates and the definintion of $\eta$
preceding the proof of Lemma \ref{lem-matrix}. From that discussion, we have
$\liminf_{k
\rightarrow \infty} \min_{{\bf u} \in (\combset{m}{k})^n}
\frac{E[\sum_{i=1}^n u_i(J_n({\bf u}))]}{\sum_{i=1}^n u_i(A)} 
= \min_{{\bf t} \in (T_3)^n} \liminf_{k
\rightarrow \infty} \frac{E[\sum_{i=1}^n u_i(J_n({\bf u}))]}{\sum_{i=1}^n u_i(A)}$, where $u_i = \eta(t_i, k)$. 
Also recall that $|T_3|
= 12$. Since $J_n$ is anonymous, to determine the approximation ratio
of $J_n$ on ${\bf u} \in (\combset{m}{k})^n$, we observe that we only
need to know the value of $k$ and the {\em fraction} of voters of each
of the possible 12 types.  In particular, fixing a type profile ${\bf
t} \in (\combset{m}{k})^n$, for each type $k \in T_3$, let $x_k$ be
the fraction of voters in ${\bf u}$ of type $k$. For convenience of
notation, we identify $T_3$ with $\{1,2,\ldots,12\}$ using the scheme
depicted in Table \ref{variables}.  Let $w_j = \lim_{k \rightarrow
\infty} \sum_{i=1}^n u_i(i)$, where $u_i = \eta(t_i,k)$, and let $p_j = \lim_{k \rightarrow \infty} \Pr[E_j]$, where $E_j$ is the event that candidate $j$ is 
elected by $J_n$ in an election with valuation profile ${\bf u}$ where $u_i = \eta(t_i, k)$. We then have 
$\liminf_{k
\rightarrow \infty} \frac{E[\sum_{i=1}^n u_i(J_n({\bf u}))]}{\sum_{i=1}^n u_i(A)} = (p_A \cdot w_A + p_B \cdot w_B + p_C \cdot w_C)/w_A.$ 
Also, from Table \ref{variables} and the definition of $J_n$, we see:
\begin{eqnarray*}
w_A & = & n (x_1 + x_2 + x_3 + x_4 + x_5 + x_9) \\
w_B & = & n (x_1 + x_5 + x_6 + x_7 + x_8 + x_{11}) \\
w_C & = & n (x_4 + x_7 + x_9 + x_{10} + x_{11}+ x_{12}) \\
p_A & = & (\cone + \ctwo/2) (x_1 + x_2 + x_3 + x_4) + (\ctwo/2)(x_5 + x_6 + x_9 + x_{10}) \\
p_B & = & (\cone + \ctwo/2) (x_5 + x_6 + x_7 + x_8) + (\ctwo/2)(x_1 + x_2 + x_{11}+x_{12}) \\
p_C & = & (\cone + \ctwo/2) (x_9 + x_{10} + x_{11} + x_{12}) + (\ctwo/2)(x_3 + x_4 + x_7 + x_8)
\end{eqnarray*}
Thus we can establish that $\mratio(J_n) > \oul$ for all $n$, by
showing that the {\em quadratic program} ``{\em Minimize $(p_A \cdot w_A +
p_B \cdot w_B + p_C \cdot w_C) - \oul w_A$ subject to $x_1 + x_2 +
\cdots + x_{12} = 1, x_1, x_2, \ldots, x_{12} \geq 0$}'', where $w_A, w_B,
w_C, p_A, p_B, p_C$ have been replaced with the above formulae using the
variables $x_i$, has a strictly positive minimum (note that the parameter $n$ appears as a multiplicative constant in the objective function and can be removed, so there is only one program, not one for each $n$). This was
established rigorously by solving the program symbolically in Maple by a facet
enumeration approach (the program being non-convex), which is easily feasible for quadratic programs of this relatively small size.
\end{proof} 

We next present a family of ordinal mechanisms $J'_n$ with $\mratio(J'_n) > \ogl$. In particular, $\agratio{O}{3} > \ogl$. The coefficents defining the mechanism $\cone$ and $\ctwo$ were again found by trial-and-error; we present more information about how later.

\begin{theorem}\label{thm-ogl}
Let $c'_1 = 0.476, c'_2=0.467$ and $d = 0.057$ and let $J_n = c'_1 \cdot \unilat{3}{n}{1} + c'_2 \unilat{3}{n}{2} + d \cdot \duple{m}{n}{\lfloor{n/2}\rfloor + 1}$. Then $\mratio(J_n) > \ogl$ for all $n$. 
\end{theorem}
\begin{proof}
The proof idea is the same as in the proof of Theorem
\ref{thm-oul}. In particular, we want to reduce proving the theorem to
solving quadratic programs. The fact that we have to deal with the
$\duple{m}{n}{\lfloor{n/2}\rfloor + 1}$, i.e., {\em random-majority}, makes this task slightly more
involved. In particular, we have to solve many programs rather than just one. We only provide a sketch, showing how to modify the proof of
Theorem \ref{thm-oul}.

As in the proof of Theorem \ref{thm-oul}, we let $w_j = \lim_{k \rightarrow
\infty} \sum_{i=1}^n u_i(i)$, where $u_i = \eta(t_i,k)$. The expressions for $w_A$ as functions of the variables $x_i$ remain the same as in that proof. 
Also, we let $p_j = \lim_{k \rightarrow \infty} \Pr[E_j]$, where $E_j$ is the event that candidate $j$ is elected by $J'_n$ in an election with valuation profile ${\bf u}$ where $u_i = \eta(t_i, k)$.  We then have
\begin{eqnarray*}
p_A & = & (c'_1 + c'_2/2) (x_1 + x_2 + x_3 + x_4) + (c'_2/2)(x_5 + x_6 + x_9 + x_{10}) + d\cdot q_A({\bf t}) \\
p_B & = & (c'_1 + c'_2/2) (x_5 + x_6 + x_7 + x_8) + (c'_2/2)(x_1 + x_2 + x_{11}+x_{12}) + d\cdot q_B({\bf t}) \\
p_C & = & (c_1' + c'_1/2) (x_9 + x_{10} + x_{11} + x_{12}) + (c'_2/2)(x_3 + x_4 + x_7 + x_8) + d\cdot q_C({\bf t}) 
\end{eqnarray*}
where $q_j({\bf t})$ is the probability that {\em random-majority} elects candidate $j$ when the type profile is ${\bf t}.$ Unfortunately, this quantity is not a linear combination of the $x_i$ variables, so we do not immediately arrive at a quadratic program.

However, we can observe that the values of $q_j({\bf t}), j=A,B,C$
depend only on the outcome of the three pairwise majority votes
between $A,B$ and $C$, where the majority vote between, say, $A$ and
$B$ has three possible outcomes: A wins, B wins, or there is a tie. In
particular, there are 27 possible outcomes of the three pairwise
majority votes. To show that $\min_{{\bf t} \in (T_3)^n} \liminf_{k
\rightarrow \infty} \frac{E[\sum_{i=1}^n u_i(J'_n({\bf
u}))]}{\sum_{i=1}^n u_i(A)} > \ogl$, where $u_i = \eta(t_i, k)$, we
partition $(T_3)^n$ into 27 sets according to the outcomes of the
three majority votes of an election with type profile ${\bf t}$ and
show that the inequality holds on all 27 sets in the partition. We
claim that on each of the 27 sets, the inequality is equivalent to a
quadratic program. Indeed, each $q_A({\bf t})$ is now a constant, and
the constraint that the outcome is as specified can be expressed as a
linear constraint in the $x_i$'s and added to the program. For
instance, the condition that $A$ beats $B$ in a majority vote can be
expressed as $x_1 + x_2 + x_3 + x_4 + x_9 + x_{10} > 1/2$ while $A$ ties
$C$ can be expressed as $x_1 + x_2 + x_3 + x_4 + x_5 + x_6 = 1/2$.
Except for the fact that these constraints are added, the program is
now constructed exactly as in the proof of Theorem
\ref{thm-oul}. Solving\footnote{To make the program amenable to
standard facet enumeration methods of quadratic programming, we changed the sharp inequalities $>$ expresssed the
majority vote constraints into weak inequalities $\geq$. Note that
this cannot decrease the cost of the optimal solution.} the programs
confirms the statement of the theorem.
\end{proof}

%We prove that the approximation ratio of any neutral, anonymous mechanism in $\rmech{O}{3}{n}$ is at most $0.624$ and then prove that the best cardinal mechanism strictly outperforms any mechanism in the afformentioned class. First we describe how to obtain approximation guarantees by formulating the problem of determining whether a mechanism achieves a given approximation ratio as a quadratic program in twelve variables, which we then solve, using quadratic programming. 

%Again, by Lemma \ref{lem-quasicomb}, for ordinal mechanisms, we only need to restrict our attention to quasi-combinatorial instances, in order to calculate the worst-case approximation ratio. As we already mentioned, for $m=3$, the set of possible types for a voter has size $|T_3|=12$. Thus, for a (infinitisimal) given value of $k$, there are twelve distinct quasi-combinatorial valuation functions we need to consider. We associate each valuation function $u_i$ with a variable $x_i$. For instance, the valuation function $u(A)=1$, $u(B)=1-\epsilon$, $u(C)=0$, where $\epsilon$ is uniquely defined by the value of $k$, is associated with variable $x_1$. This is the same as saying that a ballot of the form $\{1,1-\epsilon,0\}$ is associated with $x_1$. Variable $x_i$ denotes the fraction of voters whose quasi-combinatorial valuation function is $u_i$. For instance, if all voters' ballots are of the form $\{1,1-\epsilon,0\}$, then $x_1=1$. A list of all twelve quasi-combinatorial valuations and their corresponding variables can be found in Table \ref{variables}.

\begin{table}
\caption{Variables for types of quasi-combinatorial valuation functions with $\epsilon$ denoting $1/k$ \label{variables}}
       \begin{tabular}{ccccccccccccc}
               \textbf{Candidate}/\textbf{Variable}  &  $x_1$ & $x_2$ & $x_3$ & $x_4$ & $x_5$ & $x_6$ & $x_7$ & $x_8$ & $x_9$ & $x_{10}$ & $x_{11}$ & $x_{12}$ \\ \hline
               \noalign{\smallskip}
               A & $1$ & $1$ & $1$ & $1$ & $1-\epsilon$ & $\epsilon$  & $0$ & $0$ & $1-\epsilon$ & $\epsilon$ & $0$ & $0$  \\ 
               B & $1-\epsilon$ & $\epsilon$ & $0$ & $0$ & $1$ & $1$  & $1$ & $1$ & $0$ & $0$ & $1-\epsilon$ & $\epsilon$  \\
               C & $0$ & $0$ & $\epsilon$ & $1-\epsilon$ & $0$ & $0$  & $1-\epsilon$ & $\epsilon$ & $1$ & $1$ & $1$ & $1$  \\             
       \end{tabular}
\end{table}

We next show that $\agratio{OU}{3} \leq \ouu$ and $\agratio{O}{4} \leq \ogu$. By Lemma \ref{lem-limit}, it is enough to show that 
$\gratio{OU}{3}{n^*} \leq \ouu$ and $\gratio{O}{3}{n^*} \leq \ogu$ for some fixed $n^*$.
Therefore, the statements follow from the following theorem.

\begin{theorem}\label{thm-ou}
$\gratio{OU}{3}{23000} \leq \frac{32093343}{52579253} < \ouu$ and $\gratio{O}{3}{23000} \leq \frac{41}{64} < \ogu$.
\end{theorem}
\begin{proof}
Lemma \ref{lem-matrix} states that the two upper bounds can be proven
by showing that the values of two certain matrix games $G$ and $H$ are
smaller than the stated figures. While the two games have a reasonable
number of rows, the number of columns is astronomical, so we cannot
solve the games exactly. However, we can prove upper bounds on the
values of the games by restricting the strategy space of the column
player. Note that this corresponds to selecting a number of {\em bad
type profiles}. We have constructed a catalogue of just 5 type profiles, each with
23000 voters. Using the ``fraction
encoding'' of profiles suggested in the proof of Theorem
\ref{thm-oul}, the profiles are:
\begin{itemize}
\item{}$x_2 = 14398/23000, x_5 = 2185/23000,  x_{11}=6417/23000$.
\item{}$x_2 = 6000/23000, x_5 = 8000/23000, x_{12} = 9000/23000$.
\item{}$x_1 = 11500/23000, x_{11}=11500/23000$.
\item{}$x_2 = 9200/23000, x_5 = 4600/23000, x_{12} = 9200/23000$.
\item{}$x_2 = 13800/23000, x_{12} = 9200/23000$.
\end{itemize}
Solving the corresponding matrix games yields the
stated upper bound.
\end{proof}

While the catalogue of bad type profiles of the proof of Theorem \ref{thm-ou} suffices to prove Theorem \ref{thm-ou}, we should discuss how we arrived at this particular ``magic'' catalogue. This discussion also explains how we arrived at the ``magic'' coefficients in Theorems \ref{thm-oul} and \ref{thm-ogl}. In fact, we arrived at the catalogue and the coefficients iteratively in a joint local search process (or ``co-evolution'' process).
To get an initial catalogue, we used the fact that we had already solved the matrix games yielding the values of $\gratio{OU}{3}{n}$ and $\gratio{O}{3}{n}$, for $n = 2,3,5$. By the theorem of Shapley and Snow \cite{Shapley50}, these matrix games have optimal strategies for the column player with support size at most the number of rows of the matrices. One can think of these supports as a small set of bad type profiles for $2,3$ and $5$ voters. Utilizing that $2,3$ and $5$ all divide 1000, we scaled all these {\em up} to 1000 voters. Also, we had solved the quadratic programs of the proofs of Theorem \ref{thm-oul} and Theorem \ref{thm-ogl}, but with inferior coefficients and resulting bounds to the ones stated in this paper. The quadratic programs obtained their minima at certain type profiles. We added these entries to the catalogue, and scaled all profiles to their least common multiple, i.e. 23000. %We rounded the entries to be integer multiples of $1/3000$ and added these to the catalogue. 
Solving the linear programs of the proof of Theorem \ref{thm-ou} now gave not only an upper bound on the approximation ratio, but the optimal strategy of Player I in the games also suggested reasonable mixtures of the $\unilat{3}{n}{q}$ (in the unilateral case) and of the $\unilat{3}{n}{q}$ and {\em random-majority} (all $\duple{3}{n}{q}$ mechanisms except {\em random-majority} were assigned zero weight) to use for large $n$, making us update the coefficients and bounds of Theorem \ref{thm-oul} and \ref{thm-ogl}, with new bad type profiles being a side product. We also added by hand some bad type profiles along the way, and iterated the procedure until no further improvement was found. In the end we pruned the catalogue into a set of five, giving the same upper bound as we had already obtained.

We finally show that
$\agratio{U}{3}$ is between $\cgl$ and $\cgu$. The upper bound follows from the following proposition and Lemma \ref{lem-limit}.
\begin{proposition}
$\gratio{U}{3}{2} \leq 0.75$.
\end{proposition}
\begin{proof}
Suppose $J \in \rmech{U}{3}{2}$ has $\mratio(J) > 0.75$. By Lemma \ref{lem-naenough}, we can assume $J$ is neutral. For some $\epsilon > 0$, consider the valuation profile with $u_1(A) = u_2(A)= 1-\epsilon$, $u_1(B) = u_2(C) = 0$, and $u_1(C) = u_2(B) = 1$. As in the proof of Theorem \ref{thm-neg}, by neutrality, we must have that the probability of $A$ being elected is at most $\frac{1}{2}$. The statement follows by considering a sufficiently small $\epsilon$. 
\end{proof}
The lower bound follows from an analysis of the {\em quadratic-lottery} of Feige and Tennenholtz \cite{Feige10}. The main reason that we focus on this particular cardinal mechanism is given by the following lemma.
\begin{lemma}\label{lem-qquasicomb}
Let $J \in \mech{3}{n}$ be a convex combination of $Q_n$ and any ordinal and neutral mechanism. Then
\[  \mratio(J) = \liminf_{k \rightarrow \infty} \min_{{\bf u} \in (\combset{m}{k})^n} \frac{E[\sum_{i=1}^n u_i(J({\bf u}))]}{\sum_{i=1}^n u_i(1)}. \]
\end{lemma} 
\begin{proof}
The proof is a simple modification of the proof of Lemma \ref{lem-quasicomb}. As in that proof, for a valuation profile ${\bf u} = (u_i)$, define 
$g({\bf u}) = \frac{E[\sum_{i=1}^n u_i(J({\bf u}))]}{\sum_{i=1}^n u_i(1)}.$
We show the following equations:
\begin{eqnarray}
\mratio(J) 
& = & \inf_{{\bf u} \in \valset{3}^n} \frac{E[\sum_{i=1}^n u_i(J({\bf u}))]}{\max_{j \in M}\sum_{i=1}^n u_i(j)} \\
& = & \inf_{{\bf u} \in \valset{3}^n} g({\bf u}) \label{eq-fixed-cand2}\\ 
& = & \liminf_{k \rightarrow \infty} \min_{{\bf u} \in (\rvalset{3}{k})^n} g({\bf u}) \label{eq-discretized2}\\
& = & \liminf_{k \rightarrow \infty} \min_{{\bf u} \in (\combset{3}{k})^n} g({\bf u}) \label{eq-quasicomb2}
\end{eqnarray}

Equation (\ref{eq-fixed-cand2}), (\ref{eq-discretized2}) follows as in the proof of Lemma \ref{lem-quasicomb}.
Equation (\ref{eq-quasicomb2}) follows from the following argument.
For a profile ${\bf u} = (u_i)\in (\rvalset{3}{k})^n$, let $c_{\bf u}$ denote the number of pairs $(i,j)$ with $i$ being a voter and $j$ a candidate, for which $u_i(j)-{1/k}$ and $u_i(j) + {1/k}$ are both in $[0,1]$ and both {\em not} in
the image of $u_i$. Then, $\combset{3}{k}$ consists of exactly those ${\bf u}$ in $\rvalset{3}{k}$ for which $c_{\bf u} = 0$. To establish equation (\ref{eq-quasicomb2}), we merely have to show that for any ${\bf u} \in \rvalset{3}{k}$ for which $c_{\bf u} > 0$, there is a ${\bf u}' \in \rvalset{3}{k}$ for which
$g({\bf u}') \leq g({\bf u})$ and $c_{\bf u'} < c_{\bf u}$. 
We will now construct such ${\bf u}'$. Since $c_{\bf u} > 0$, there is a pair $(i,j)$ so that $u_i(j)-{1/k}$ and $u_i(j) + {1/k}$ are both in $[0,1]$ and both not in the image of $u_i$. Let $\ell_-$ be the smallest integer value so that
$u_i(j) - {\ell/k}$ is not in the image of $u_i$, for any integer $\ell \in \{\ell_-, \ldots, j-1\}$. Let $\ell_+$ be the largest integer value so that
$u_i(j) + {\ell/k}$ is not in the image of $u_i$, for any integer $\ell \in \{j+1, \ldots, \ell_+\}$. We can define a valuation function $u^x \in \valset{m}$ for any $x \in [-\ell_-/k; \ell_+/k]$ 
as follows: $u^x$ agrees with $u_i$ except on $j$, where we let $u^x(j) = u_i(j) + x$. Let ${\bf u}^x = (u^x, u_{-i})$. Now consider the function $h: x \rightarrow g({\bf u}^x)$. Since $J$ is a convex combination of {\em quadratic-lottery} and a neutral ordinal mechanism, we see by inspection of the definition of the function $g$, that $h$ on the domain $[-\ell_-/k; \ell_+/k]$  is the quotient of two quadratic polynomials where the numerator has 
second derivative being a negative constant and the denominator is postive throughout the interval.
This means that $h$ attains its minimum at either $\ell_-/k$ or at $\ell_+/k$.
 In the first case, we let ${\bf u}' = {\bf u}^{\ell_-/k}$ and in the second, we let ${\bf u}' = {\bf u}^{\ell_+/k}$. This completes the proof.
\end{proof}

%We now calculate the approximation ratio of the Quadratic Lottery. The result is rather intriguing.

\begin{theorem}\label{thm-cardthree}
The limit of the approximation ratio of $Q_n$ as $n$ approaches infinity, is exactly the golden ratio, i.e., $(\sqrt{5}-1)/2  \approx 0.618$. Also, let $J_n$ be the mechanism for $n$ voters that selects {\em random-favorite} with probability $29/100$ and {\em quadratic-lottery} with probability $71/100$. Then, $\mratio(J_n) > \frac{33}{50}=0.660$.
\end{theorem}  
\begin{proof} (sketch)
Lemma \ref{lem-qquasicomb} allows us to proceed completely as in the proof of Theorem \ref{thm-oul}, by constructing and solving appropriate quadratic programs. As the proof is a straightforward adaptation, we leave out the details.
\end{proof}

Mechanism $J_n$ of Theorem \ref{thm-cardthree} achieves an approximation ratio strictly better than $0.64$. In other words, the best truthful cardinal mechanism for three candidates strictly outperforms all ordinal ones.

\section{Conclusion}
\label{sec-conc}
By the statement of Lee \cite{Lee14}, mixed-unilateral mechanisms are asymptotically no better than ordinal mechanisms. Can a cardinal mechanism which is not mixed-unilateral beat this approximation barrier? Getting upper bounds on the performance of general cardinal mechanisms is impaired by the lack of a
characterization of cardinal mechanisms a la Gibbard's. Can we adapt the proof of Theorem \ref{thm:anyupper} to work in the general setting without ties? 
For the case of $m=3$, can we close the gaps for ordinal mechanisms and for mixed-unilateral mechanisms? How well can cardinal mechanisms do for $m=3$? Theorem \ref{thm-dm} holds for $m=3$ as well, but perhaps we could prove a tighter upper bound for cardinal mechanisms in this case.

%Electronic Appendix
%\appendixhead{url}

% Bibliography

\bibliographystyle{plain}
\bibliography{voting}

\begin{thebibliography}{10}

\bibitem{Barbera78}
Salvador Barbera.
\newblock Nice decision schemes.
\newblock In Leinfellner and Gottinger, editors, {\em Decision theory and
  social ethics}. Reidel, 1978.

\bibitem{Barbera79}
Salvador Barbera.
\newblock Majority and positional voting in a probabilistic framework.
\newblock {\em The Review of Economic Studies}, 46(2):379--389, 1979.

\bibitem{Barbera10}
Salvador Barbera.
\newblock Strategy-proof social choice.
\newblock In K.~J. Arrow, A.~K. Sen, and K.~Suzumura, editors, {\em Handbook of
  Social Choice and Welfare}, volume~2, chapter~25. North-Holland: Amsterdam,
  2010.

\bibitem{Barbera98}
Salvador Barbera, Anna Bogomolnaia, and Hans van~der Stel.
\newblock Strategy-proof probabilistic rules for expected utility maximizers.
\newblock {\em Mathematical Social Sciences}, 35(2):89--103, 1998.

\bibitem{Boutilier}
Craig Boutilier, Ioannis Caragiannis, Simi Haber, Tyler Lu, Ariel~D. Procaccia,
  and Or~Sheffet.
\newblock Optimal social choice functions: a utilitarian view.
\newblock In {\em Proceedings of the 13th ACM Conference on Electronic
  Commerce}, pages 197--214. ACM, 2012.

\bibitem{Feige10}
Uriel Feige and Moshe Tennenholtz.
\newblock Responsive lotteries.
\newblock In {\em SAGT 2010, Proceedings}, volume 6386 of {\em Lecture Notes in
  Computer Science}, pages 150--161. Springer, 2010.

\bibitem{Aris14}
Aris Filos-Ratsikas, S{\o}ren Kristoffer~Stiil Frederiksen, and Jie Zhang.
\newblock {Social welfare in one-sided matchings: Random priority and beyond}.
\newblock In {\em Proceedings of the 7th Symposium of Algorithmic Game Theory.
  To appear}. Springer, 2014.

\bibitem{Freixas84}
Xavier Freixas.
\newblock A cardinal approach to straightforward probabilistic mechanisms.
\newblock {\em Journal of Economic Theory}, 34(2):227 -- 251, 1984.

\bibitem{Gibbard73}
Allan Gibbard.
\newblock Manipulation of voting schemes: A general result.
\newblock {\em Econometrica}, 41(4):587--601, 1973.

\bibitem{Gibbard77}
Allan Gibbard.
\newblock Manipulation of schemes that mix voting with chance.
\newblock {\em Econometrica}, 45(3):665--81, 1977.

\bibitem{Gibbard78}
Allan Gibbard.
\newblock Straightforwardness of game forms with lotteries as outcomes.
\newblock {\em Econometrica}, 46(3):595--614, 1978.

\bibitem{Lee14}
Anthony~Sinya Lee.
\newblock Maximization of relative social welfare on truthful voting scheme
  with cardinal preferences.
\newblock Working manuscript, 2014.

\bibitem{Nisan07}
Noam Nisan.
\newblock {\em Algorithmic Game Theory}, chapter 9: Introduction to Mechanism
  Design (for Computer Scientists), pages 209--241.
\newblock Cambridge University Press, New York, NY, USA, 2007.

\bibitem{NisanRonen}
Noam Nisan and Amir Ronen.
\newblock Algorithmic mechanism design (extended abstract).
\newblock In {\em Proceedings of the thirty-first annual ACM symposium on
  Theory of computing}, pages 129--140. ACM, 1999.

\bibitem{Procaccia10}
Ariel~D. Procaccia.
\newblock Can approximation circumvent {Gibbard}-{Satterthwaite}?
\newblock In {\em AAAI 2010, Proceedings}. AAAI Press, 2010.

\bibitem{PT}
Ariel~D. Procaccia and Moshe Tennenholtz.
\newblock Approximate mechanism design without money.
\newblock In {\em Proceedings of the 10th ACM conference on Electronic
  commerce}, pages 177--186. ACM, 2009.

\bibitem{Satterthwaite75}
Mark~Allen Satterthwaite.
\newblock Strategy-proofness and {Arrow}'s conditions: Existence and
  correspondence theorems for voting procedures and social welfare functions.
\newblock {\em Journal of Economic Theory}, 10(2):187--217, 1975.

\bibitem{Shapley50}
L.~S. Shapley and R.~N. Snow.
\newblock Basic solutions of discrete games.
\newblock In {\em Contributions to the Theory of Games}, number~24 in Annals of
  Mathematics Studies, pages 27--35. Princeton University Press, 1950.

\bibitem{Yao77}
Andrew Chi-Chi Yao.
\newblock Probabilistic computations: Toward a unified measure of complexity.
\newblock In {\em Foundations of Computer Science, 1977., 18th Annual Symposium
  on}, pages 222--227. IEEE, 1977.

\bibitem{Zeckhauser73}
R.~Zeckhauser.
\newblock Voting systems, honest preferences and {Pareto} optimality.
\newblock {\em The American Political Science Review}, 67:934--946, 1973.

\end{thebibliography}

%Received environment
%\received{Month Year}{Month Year}{Month Year}

\end{document}